\documentclass{article}
\usepackage{amsmath,amsfonts,amssymb,amsthm}
\usepackage[round]{natbib}
\usepackage{color}
\usepackage{comment}
\usepackage{multirow}
\usepackage{verbatim}
\usepackage{bbm}
\usepackage{booktabs}
\usepackage{float}
\usepackage{graphicx}
\usepackage{paralist} 
\allowdisplaybreaks

\newcommand{\wh}{\widehat}
\newcommand{\wt}{\widetilde}
\newcommand{\real}{\mathbb{R}}

\newcommand{\bsx}{\boldsymbol{x}}

\newcommand{\bsX}{\boldsymbol{X}}

\newcommand{\e}{\mathbb{E}}
\newcommand{\var}{\mathrm{var}}
\newcommand{\cov}{\mathrm{cov}}

\newcommand{\rct}{\mathcal{R}}
\newcommand{\odb}{\mathcal{O}}

\newcommand{\err}{\varepsilon}

\renewcommand{\le}{\leqslant}
\renewcommand{\ge}{\geqslant}

\newcommand{\tran}{\mathsf{T}}

\newcommand{\simiid}{\stackrel{\mathrm{iid}}\sim}

\DeclareMathOperator*{\argmin}{argmin}

\newcommand{\edelt}{\e_{\delta}}   
\newcommand{\vdelt}{\var_{\delta}} 
\newcommand{\bdelt}{\bias_{\delta}}

\newcommand{\ok}{\mathrm{ok}}
\newcommand{\okt}{\mathrm{okt}}
\newcommand{\okc}{\mathrm{okc}}

\renewcommand{\it}{it}
\newcommand{\ic}{ic}

\newcommand{\yit}{Y_{\it}}
\newcommand{\yic}{Y_{\ic}}

\newcommand{\wi}{W_i}
\newcommand{\wit}{W_{\it}}
\newcommand{\wic}{W_{\ic}}

\newcommand{\bern}{\mathrm{Bern}}
\newcommand{\mse}{\mathrm{MSE}}
\newcommand{\bias}{\mathrm{Bias}}
\newcommand{\sets}{\mathcal{S}}

\newcommand{\phm}{\phantom{-}}

\newcommand{\dnorm}{\mathcal{N}}
\newcommand{\sd}{\mathrm{SD}}

\newtheorem{proposition}{Proposition}
\newtheorem{corollary}{Corollary}
\newtheorem{theorem}{Theorem}

\theoremstyle{definition} 
\newtheorem{assumption}{Assumption}

\begin{document}
\begin{titlepage}
	\centering
	{\scshape\LARGE Propensity Score Methods for Merging Observational and Experimental Datasets  \par}
	\vspace{1cm}
	{\Large\itshape 
Evan Rosenman (Stanford University)\\ 
Michael Baiocchi (Stanford University)\\
Hailey Banack (University at Buffalo) \\
Art B. Owen (Stanford University)\\[1.5ex]
\rm
\par}
	\vfill
	{\large October 2018} 
\end{titlepage}

\newpage
\tableofcontents
\newpage

\begin{abstract}
This project considers how one might augment a limited amount of data from randomized controlled trial (RCT) with more plentiful data from an observational database (ODB), in order to estimate a causal effect. 
In our motivating setting, the ODB has better external validity,
while the RCT has genuine randomization.
We work with strata defined by the propensity score in the ODB.
Subjects from the RCT are placed in strata defined by
the propensity they would have had, had they been in the ODB.
Our first method simply spikes the RCT data into their corresponding ODB strata.
Our second method takes a data driven convex combination of the ODB and RCT treatment effect 
estimates within each stratum. 
Using the delta method and simulations we show that the spike-in method works best when
the RCT covariates are drawn from the same distribution as in the ODB.
Our convex combination method is more robust than the spike-in to covariate-based inclusion criteria
that bias the RCT data.
We apply our methods to data from the Women's Health Initiative, a study of thousands of postmenopausal women 
which has both observational and experimental data on hormone therapy (HT).
Using half of the RCT to define a gold standard, we find that a version of the
spiked-in estimate yields stable estimates of the causal impact of HT
on coronary heart disease.
\end{abstract}

\section{Introduction}\label{sec:introduction}

The increasing availability of large, observational datasets poses opportunities and challenges for statistical methodologists. These datasets often ``contain detailed information about the target population of interest," meaning they could have great utility for estimating the causal effects of a proposed intervention, such as a public health initiative \citep{Hartman_fromsate}. Yet assignment to the treatment is non-random in these data, making standard methods prone to misstate the treatment effect. 

Randomized control trials (RCTs) make causal inference substantially easier, because the researcher controls assignment of the intervention. Yet RCTs present their own challenges. A commonly raised concern is that ``estimates from RCTs $\ldots$ may lack external validity" \citep{Hartman_fromsate} due to the sampling scheme used to enroll participants. Moreover, in cases where the treatment effect varies by subpopulation, ``experiments have to be very large and, in general, prohibitively costly" whereas ``observational data is often available in much larger quantities" \citep{DBLP:journals/corr/PeysakhovichL16}. 

These issues motivate a hybrid approach, which makes use of the availability, size, and representativeness of observational data as well as the randomization inherent in RCT data. Implicitly, any such approach will also involve combining biased and unbiased estimators -- a problem with substantial precedent in the literature \citep{mundlak1978pooling, green1991james}. The fundamental tool used in our approach is the propensity score: the estimated conditional probability of exposure to the intervention, given observed covariates. 

\cite{rosenbaum1984reducing} showed that comparing treated individuals against control individuals for whom the propensity score is approximately equal yields a substantial reduction in bias. The propensity score is widely used in analysis of observational datasets, as comparing test and control units with similar propensity scores ``tends to balance observed covariates that were used to construct the score" \citep{doi:10.1093/oxfordjournals.aje.a010011}. It has been less widely used in the context of RCTs. But propensity-based methods ``have been shown to be useful even in randomized control settings, where the assignment mechanism is known and independent of the covariates" \citep{BIOM:BIOM1364}, because they correct for chance imbalance in covariates. 

In using the propensity score, we make several simplifying assumptions: we assume a strongly ignorable treatment assignment (SITA), as defined by \cite{rosenbaum1984reducing}; and we allow for a heterogeneous treatment effect but assume that it varies only as a function of the propensity. Both are strong assumptions. The existence of a prospectively designed RCT in our setting may yield some insights into the factors affecting treatment assignment in the ODB, lending some additional plausibility to the SITA assumption. Yet it remains hard to defend in practice, and a natural extension of this work will involve engaging with unmeasured confounders.

The latter assumption is retained in its strictest form for simplicity, but can be greatly weakened by sub-stratifying on additional features assumed to be related to the heterogeneity of the treatment effect. A practical example is provided in Section~\ref{subsec:outcomeModeling}. The treatment effect for subjects in the RCT is assumed to be the same function of the propensity that holds in the ODB.  The distribution of covariates in the RCT can however be much different from the ODB due to factors such as varying enrollment criteria.



We assume that the covariate distribution in the ODB is the same as that of the target population.
We use data from the RCT to provide much-needed control subjects in the strata with a high propensity
in the ODB as well as test group subjects in the strata with a low propensity in the ODB. Note that this is a relatively simple case of causal transportability, as, per the discussion above, we assume the propensity score is a causal modifier and that other modifiers are known via subject matter expertise. We are also not engaging with the question of selection diagrams, and implicitly making strong invariance assumptions \citep{pearl2014external}. 

Our two main estimators are a spiked-in estimator that simply merges the ODB and RCT within each stratum
and a ``dynamically weighted'' estimator that mimics the best possible convex combination of 
estimates from the two populations, within each stratum.





The remainder of the paper is organized as follows. In Section~\ref{sec:notation},
we define our notation, assumptions, estimand and estimators, including the
spiked-in and dynamic weighting estimators that we propose.
We work in the potential outcomes framework, in which treatment effects are ratio
estimators due to the random numbers of subjects in each condition.  For the large sample
sizes of interest to us, delta method approximations to the mean and variance are
accurate enough. Section~\ref{sec:delta} presents delta method estimates of the within-stratum
bias and variance for our estimators.  There we see theoretically that the spiked-in estimator
can have an enormous bias if the covariates in the RCT do not follow the same distribution
as those of the ODB.

Section~\ref{sec:simulations} gives some numerical illustrations of our method for an ODB of size $5{,}000$
and an RCT of size $200$.  When the RCT covariates follow the ODB's distribution, then the spiked-in
estimator brings a large reduction in mean squared error over the ODB-only estimate. If, however, enrollment criteria bias the RCT (i.e. the treatment effect varies by subpopulation and the RCT covariate distribution differs substantively from that of the ODB), then the spiked-in estimator can be much worse than the ODB-only estimate. In either case, the dynamic weighted estimator brings an improvement over the ODB-only estimate. Section~\ref{sec:whi} introduces data from the Women's Health Initiative (WHI). An overview of the components of the WHI and the data collection process is provided in Section~\ref{sec:whi}, while the subsequent sections detail how a propensity model is fit to the WHI observational component and a ``gold standard" estimate of the causal effect is derived. A variant of the spiked-in estimator is introduced in Section~\ref{subsec:outcomeModeling}. That variant refines propensity strata via a prognostic score, leading to a ``dual spiked'' estimator.  All the estimators are compared via bootstrap simulations in \ref{subsec:results}, and the dual spiked estimate is most accurate for the WHI data. Section~\ref{sec:conclusions} summarizes our conclusions and an Appendix contains two of
the lengthier proofs.

\section{Notation, assumptions and estimators}\label{sec:notation}

Some subjects belong to the randomized controlled trial (RCT)
and others to the observational database (ODB).
We assume that no subject is in both data sets.
We write $i\in\rct$ if subject $i$ is in the RCT and $i\in\odb$ otherwise. 
Subject $i$ has an outcome $Y_i\in\real$
and some covariates that we encode in
the vector $\bsx_i\in\real^d$.
Subject $i$ receives either the test condition or
the control condition.  

The condition of subject $i$ is given
by a treatment variable $W_i\in\{0,1\}$ where $W_i=1$
if subject $i$ is in the test condition (and $0$ otherwise).
Some formulas simplify when we can use parallel
notation for both test and control settings.
Accordingly we introduce
$\wit=W_i$ and $\wic=1-W_i$.
Other formulas look better when focused
on the test condition.  For instance, letting $p_{it}=\Pr(\wit=1)$
and $p_{ic}=\Pr(\wic=1)$,
the expression $p_{it}(1-p_{it})$ is immediately recognizable
as a Bernoulli variance and is preferred to $p_{it}p_{ic}$.



\subsection{Model}
We adopt the potential outcomes framework of Neyman and Rubin.
See \cite{rubin1974estimating}.
Subject $i$ has two potential outcomes, $\yit$ and $\yic$, corresponding
to test and control conditions respectively.
Then $Y_i = \wit\yit +\wic\yic$.
The potential outcomes $(\yit,\yic)$ are non-random and we will
assume that they are bounded.
We work conditionally on the observed
values of covariates and so $\bsx_i$ are also non-random.

All of the randomness comes from the treatment variables $W_i$.
We use the notation $\bern(p)$ for Bernoulli random variables
taking the value $1$ with probability $p$ and $0$ with probability
$1-p$.  The ODB and RCT differ in how the $W_i$ are distributed.

\begin{assumption}[ODB sampling]\label{assumption:odb}
If $i\in\odb$, then $W_i\sim\bern(p_i)$ independently where 
$p_i = e(\bsx_i)$ with $0<p_i<1$. 
\end{assumption}

The function $e(\cdot)$ in 
Assumption~\ref{assumption:odb} is a propensity. 
Because the propensity depends only on $\bsx$,
and is never $0$ or $1$,  the ODB 
has a strongly ignorable treatment assignment 
\citep{rosenbaum1984reducing}. 
Because the $W_i$ are independent, 
the outcome for subject $i$ is unaffected 
by the treatment $W_{i'}$ for any subject $i'\ne i$. 
That is, our model for the ODB 
satisfies the stable unit treatment value assumption or SUTVA 
\citep{Imbens:2015:CIS:2764565}.

\begin{assumption}[RCT sampling]\label{assumption:rct}
If $i\in\rct$, then $W_i\sim\bern(p_r)$ independently for a common
probability $0<p_r<1$. 
\end{assumption}

The RCT will commonly have $p_r=1/2$ but we do not assume this.
Our RCT model also has a strongly ignorable treatment assignment
and it too satisfies the SUTVA.
We additionally assume that the ODB is independent of the RCT.

\subsection{Stratification}\label{sec:stratification}

\begin{table}
\centering
\begin{tabular}{lllll}
\toprule
Symbols & Meaning\\
\midrule
$i$, $k$, $\omega_k$ & Subject and stratum indices, stratum weights\\
$\odb$, $\rct$, $\odb_k$, $\rct_k$ & ODB and RCT subject sets and strata\\
$\bsx_i$, $e(\bsx_i)$ & Covariates and ODB propensities\\
$Y_{it}$, $Y_{ic}$, $Y_i$ & Test, control and observed responses\\
$W_{it}$, $W_{ic}$, $W_i$ & Test, control and observed indicators. $W_i\equiv W_{it}$\\
$\tau_{ok}=\tau_{rk}=\tau_k$ & Stratum treatment effects: ODB, RCT, merged\\
$n_{ok}$, $n_{rk}$, $n_k$ & Stratum sample sizes: ODB, RCT, merged\\
$\hat\tau_{ok}$, $\hat\tau_{rk}$ 
& ODB and RCT estimates of $\hat\tau_k$\\
$\hat\tau_{sk}$, $\hat\tau_{wk}$, $\hat\tau_{dk}$ 
& Spiked, weighted, dynamic estimates of $\hat\tau_k$\\
$\mu_{okt}$, $\mu_{okc}$ & Average potential responses by ODB stratum\\
$\mu_{rkt}$, $\mu_{rkc}$ & Average potential responses by RCT stratum\\
$p_{okt}$, $p_{rkt}$ & Average propensities by ODB and RCT strata\\
$p_{rkc}$, $p_{rkc}$ & One minus average propensities\\
$s_{okt}$, $s_{okc}$ & Response-propensity covariances in the ODB\\
\bottomrule
\end{tabular}
\caption{Summary of notation used.}
\end{table}


Our comparison of treatment versus control is based on stratification by
propensity as described by
\cite{Imbens:2015:CIS:2764565}.
This is one of several matching strategies mentioned in \cite{Stuart07bestpractices}.

We use $K$ strata defined by propensity intervals.
For the ODB these are
$$
\odb_k = \Bigl\{ i\in\odb\bigm| \frac{k-1}K < e(\bsx_i) \le \frac{k}K\Bigr\},
\quad k=1,\dots,K.
$$
The RCT  is similarly stratified via
$$
\rct_k = \Bigl\{ i\in\rct\bigm| \frac{k-1}K < e(\bsx_i) \le \frac{k}K\Bigr\},
\quad k=1,\dots,K.
$$
Note that the RCT is stratified according to the propensity
that those observations would have had, had they been in the ODB.
\cite{Imbens:2015:CIS:2764565} suggest strata containing approximately
equal numbers of observations. We have instead given them equal
sized propensity ranges.  In practice, some small strata might have to
be merged together.
Sometimes we refer to strata as `bins'.

We work here as though the propensity function $e$ is known
exactly. In practice,  $e$ will be replaced by an
estimated propensity fit to the ODB.  We suppose that
the ODB is large enough to obtain a good propensity estimate.
Under Assumption~\ref{assumption:odb}, 
the true propensity is a function of the observed variables~$\bsx_i$.

The sample sizes of the ODB and RCT are $n_o$ and $n_r$
respectively.  Ordinarily $n_o\gg n_r$.
The ODB and RCT sample sizes within stratum $k$ are $n_{ok}$ and $n_{rk}$. 
The within-stratum average treatment effects are
\begin{align}\label{eq:tauk}
\tau_{ok} = \frac1{n_{ok}}\sum_{i\in\odb_k}\yit-\yic\quad\text{and}\quad
\tau_{rk} = \frac1{n_{rk}}\sum_{i\in\rct_k}\yit-\yic,
\end{align}
where means over empty strata are taken to be $0$ in~\eqref{eq:tauk}.

\begin{assumption}\label{assumption:treatmenteffect}
For $k=1,\dots,K$,
if $n_{ok}>0$ and $n_{rk}>0$ then $\tau_{ok}=\tau_{rk}$ and 
we call their common value $\tau_k$. 
\end{assumption}

Assumption~\ref{assumption:treatmenteffect} leaves $\tau_k$
undefined when $\min(n_{ok},n_{rk})=0$.
If only one of $n_{ok}$ and $n_{rk}$
is positive then we take its treatment effect for $\tau_k$. If both are $0$,
then we will not need $\tau_k$.

\begin{assumption}\label{assumption:strongtreatmenteffect}
For all $i\in\odb_k\cup\rct_k$, $\yit-\yic=\tau_k$.
\end{assumption}


Assumption~\ref{assumption:strongtreatmenteffect} is an idealization that simplifies some derivations, and
we need it in one instance to estimate a quantity that depends on both potential outcomes of a single subject.
In some of our simulations we will relax that assumption to make
$Y_{it} - Y_{ic}$ constant for all units $i$ at a fixed propensity score, rather than within a stratum.  
 \cite{xie2012estimating} have argued for analyzing the pattern of treatment effects solely as a function of the propensity score, an approach taken by a number of social science researchers \citep{brand2011impact, dellaposta2013heterogeneous}. 
Because our strata are based on propensity,  Assumption~\ref{assumption:strongtreatmenteffect} 
is very nearly true under the model of \cite{xie2012estimating}.

Assumption~\ref{assumption:strongtreatmenteffect} can be made more realistic by stratifying on both the propensity score and a `prognostic score' predicting the potential outcome for subjects in the control group.  We do this in Section~\ref{subsec:outcomeModeling} in the context of the WHI data.

\subsection{Estimators}

Our estimand is a  global average treatment effect defined by
\[ \tau = \sum_{k = 1}^K \omega_k \tau_k \] 
for weights $\omega_k\ge0$ with $\sum_{k=1}^K\omega_k=1$.
The weights can be chosen to match population characteristics.
Our choice is to take $\omega_k = n_{ok}/n_o$ which is reasonable
when the ODB represents the target population of interest.
With this choice, $\omega_k=0$ whenever $n_{ok}=0$ and
we have a well defined $\tau_k$ for every stratum that contributes to $\tau$.
We may still have $n_{rk}=0$ for some strata with $\omega_k>0$.

We now introduce some estimators designed to make use of the advantages of both the RCT and the ODB data. 
Our estimators all take the form $\sum_k\omega_k\hat\tau_k$ for different
within-stratum estimates $\hat\tau_k$.

Our two simplest proposed estimators each use just one of the two populations. 
The ODB-only estimate of the treatment effect in stratum $k$ is
\begin{align}\label{eq:deftauok}
 \hat \tau_{ok} = \frac{\sum_{i \in \mathcal{O}_k} W_{it} Y_{it} }{\sum_{i \in \mathcal{O}_k} W_{it}} - \frac{\sum_{i \in \mathcal{O}_k} W_{ic} Y_{ic} }{\sum_{i \in \mathcal{O}_k} W_{ic} }.
\end{align}
Then $\hat \tau_{o} = \sum_k\omega_k\hat\tau_{ok}$.
A potential problem with $\hat\tau_{o}$ is that small values
of $k$, corresponding to the left-most bins, have subjects
with small propensity values.
Then $\odb_k$ may contain very few observations with $W_{it}=1$
and $\hat\tau_{ok}$ may have high variance.
Similarly for large $k$, $\odb_k$ may contain very few observations with $W_{ic}=1$
which again leads to high variance.
That is, the `edge bins' can have very skewed sample sizes causing problems for $\hat\tau_o$.

The ODB estimate~\eqref{eq:deftauok} is a difference of ratio estimators, because the
denominators are random. We will see in Section \ref{sec:delta} that there can also be a severe bias in the edge bins.

An analogous RCT-only estimator is $\hat\tau_r =\sum_k\omega_k\hat \tau_{rk}$ where
\begin{align}\label{eq:deftaurk}
 \hat \tau_{rk} = \frac{\sum_{i \in \mathcal{R}_k}W_{it} Y_{it} }{\sum_{i \in \mathcal{R}_k} W_{it}} - \frac{\sum_{i \in \mathcal{R}_k}W_{ic} Y_{ic} }{\sum_{i \in \mathcal{R}} W_{ic}}.
\end{align}
Because the RCT assigns treatments with constant probability, 
the edge bins have less skewed treatment outcomes.
However, because the RCT is small, we may find that several
of the strata have very small sample sizes $n_{rk}$.

Our first hybrid estimator is $\hat \tau_s = \sum_k\omega_k\hat\tau_{sk}$, where
\begin{align}\label{eq:deftausk}
\hat\tau_{sk} =\frac{\sum_{i \in \mathcal{O}_k} W_{it} Y_{it} + \sum_{i \in \mathcal{R}_k}W_{it} Y_{it} }{\sum_{i \in \mathcal{O}_k} W_{it} + \sum_{i \in \mathcal{R}_k}W_{it} } - \frac{\sum_{i \in \mathcal{O}_k} W_{ic} Y_{ic} + \sum_{i \in \mathcal{R}_k}W_{ic} Y_{ic}  }{\sum_{i \in \mathcal{O}_k} W_{ic} + \sum_{i \in \mathcal{R}_k}W_{ic}}.
\end{align}
The RCT data are `spiked' into the ODB strata. 
This spiked-in estimator can improve upon the ODB estimator
by increasing the number of treated units in the low propensity
edge bins and increasing the number of control units in the high propensity
edge bins. Even a small number of such balancing observations can be extremely valuable.

The spiked-in estimator is not a convex combination of $\hat \tau_{ok}$ and $\hat \tau_{rk}$, because the pooling is first done among the test and control units. Our final two estimators \emph{are} constructed as convex combinations of $\hat \tau_{ok}$ and $\hat \tau_{rk}$.

The weighted average estimator $\hat \tau_{w}$ uses
\begin{align}\label{eq:deftauwk}
\hat \tau_{wk} = \lambda_k\hat \tau_{ok} + (1-\lambda_k)\hat \tau_{rk}, 
\quad\text{where}\quad\lambda_k = \frac{n_{ok}}{n_{ok} + n_{rk}}.
\end{align}
It weights $\hat \tau_{rk}$ and $\hat \tau_{ok}$ according to the number of data points involved in each estimate. 

Our final estimator is a ``dynamic weighted average" $\hat \tau_{d}$.
It uses weights for $\hat \tau_{rk}$ and $\hat \tau_{ok}$ that are estimated from the 
data. Those weights are chosen to minimize an estimate of mean squared error (MSE) derived using the delta method in the following section.
While the precise form of this estimator will be discussed next, we can observe its approximate optimality via the following result, recalling that the RCT estimator will in general be unbiased. 

\begin{proposition}\label{proposition:invMSE}
Let $\hat\phi_1$ and $\hat\phi_2$ be independent estimators
of a common quantity $\phi$, with bias, variance and mean squared errors,
$\bias(\hat\phi_1) \in (-\infty, \infty) $, $\bias(\hat\phi_2) = 0$, $\var(\hat\phi_j)$ and $\mse(\hat\phi_j)\in(0,\infty)$ for $j=1,2$.
For $c\in\real$, let $\hat\phi_{c} = c \hat \phi_1 + (1 - c)  \hat \phi_2$.
Then
\[ c_*\equiv 
\argmin_c \mse(\hat\phi_c) = 
\frac{\var(\hat \phi_2)}{\mse(\hat \phi_1) + \var(\hat \phi_2) }.
\] 
This linear combination has
\begin{align}
\begin{split}
\bias(\hat\phi_{c_*}) &=
\frac{\bias(\hat\phi_1)\mse(\hat \phi_2)}{\mse(\hat \phi_1) + \mse(\hat \phi_2) },\\
\var(\hat\phi_{c_*}) &=
c_*^2\var(\hat\phi_1) +(1-c_*)^2\var(\hat\phi_2),\quad\text{and}\\
 \mse(\hat \phi_{c_*}) &= \frac{\mse(\hat \phi_1) \var(\hat \phi_2) }{\mse(\hat \phi_1) + \var(\hat \phi_2)}.
\end{split}
\end{align}
\end{proposition}
\begin{proof}
Independence of the  $\hat\phi_j$ yields $\var(\hat\phi_c) = 
c^2\var(\hat\phi_1)
+(1-c)^2\var(\hat\phi_2)$ while linearity of expectation yields $\bias(\hat\phi_c) = c \bias(\hat \theta_1)$. Optimizing $\mse(\hat \phi_c)$ over $c$ yields the result. 
\end{proof}

\section{Delta method results}\label{sec:delta}

In this section we develop some delta method moment approximations.
Let $\bsX$ be a random vector with mean $\mu$
and a finite covariance matrix.
Let $f$ be a function of $\bsX$ that is twice differentiable
in an open set containing $\mu$
and let $f_1$ and $f_2$ be first and second order
Taylor approximations to $f$ around $\mu$.
Then the delta method mean and variance of $f(\bsX)$ are
$$\edelt(f(\bsX))=\e(f_2(\bsX))\quad
\text{and}\quad
\vdelt( f(\bsX))=\var(f_1(\bsX))$$
respectively.

Sometimes, to combine estimates,
we will need a delta method mean for a weighted
sum of those estimates.  
We will also need a delta method variance for 
a weighted sum of independent random variables.
We use the following natural expressions
\begin{align}
\edelt\Biggl( \sum_j\lambda_j\hat\tau_j\Biggr) 
&= \sum_j\lambda_j\edelt(\hat\tau_j) \label{eq:edeltcombo}\\
\vdelt\Biggl( \sum_j\lambda_j\hat\tau_j\Biggr) 
&= \sum_j\lambda_j^2\vdelt(\hat\tau_j),\quad\text{for independent $\hat\tau_j$} \label{eq:vdeltcombo}
\end{align}
without making recourse to Taylor approximations.

\subsection{Population quantities}\label{sec:usefulQuantities}
We will study our estimators in terms of some population
quantities. These involve some  unobserved
values of  $\yit$ or $\yic$.
For instance, the test and control stratum averages in the ODB are
\[ \mu_{okt} = \frac{\sum_{i \in \mathcal{O}_k} Y_{it}}{n_{ok}}
\quad\text{and}\quad
 \mu_{okc} = \frac{\sum_{i \in \mathcal{O}_k} Y_{ic}}{n_{ok}}\]
and it is typical that both of these are unobserved.
Corresponding values for the RCT are $\mu_{rkt}$ and $\mu_{rkc}$.

When we merge ODB and RCT strata we will have to consider
a kind of skew in which the within-stratum mean responses above
differ between the two data sets.
To this end, define
\[ 
\Delta_{kt} = \mu_{okt} - \mu_{rkt}\quad\text{and}\quad 
\Delta_{kc} = \mu_{okc} - \mu_{rkc}.\]
Under Assumption \ref{assumption:treatmenteffect}, 
$\Delta_{kt} = (\tau_k+\mu_{okc}) 
-(\tau_k+\mu_{rkc})=\Delta_{kc}$. 
We will use $\Delta_k = \Delta_{kt} = \Delta_{kc}$. 


Now we define several other population quantities.
Let $\sets$ be a finite non-empty set of $n=n(\sets)$ indices 
such as one of our strata $\odb_k$ or $\rct_k$. 
For each $i\in\sets$, let $(\yit,\yic)\in[-B,B]^2$ be a pair of bounded potential 
outcomes and let $W_i=W_{it}$ be independent $\bern(p_i)$
random variables and let $\wic =1-\wit$.
Some of our results add the condition that all $p_i\in[\epsilon,1-\epsilon]$ for some $\epsilon>0$.

For $\sets$ so equipped, we define average responses 
\begin{align}\label{eq:mufors}
\mu_t=\mu_t(\sets) & = \frac1n\sum_{i\in\sets}\yit 
\quad\text{and}\quad\mu_c=\mu_c(\sets)  = \frac1n\sum_{i\in\sets}\yic. 
\end{align}
For example, $\mu_{okt}$ above is $\mu_t(\odb_k)$.
We use average treatment probabilities 
\begin{align}\label{eq:pfors}
p_t=p_t(\sets) & = \frac1n\sum_{i\in\sets} p_i 
\quad\text{and}\quad p_c = p_c(\sets) = 1-p_t(\sets). 
\end{align}
These become $p_{okt}$, $p_{okc}$, $p_{rkt}$ and $p_{rkc}$ in a natural notation when
$\sets$ is $\odb_k$ or~$\rct_k$.

The above quantities are averages over $i$ uniformly distributed in $\sets$
as distinct from expectations with respect to random $W_i$.
We also need some covariances of this type between response and propensity 
values,
\begin{align}\label{eq:prcovfors}
\begin{split}
s_t=s_t(\sets) &= \frac1n{\sum_{i \in \mathcal{S}} Y_{it} p_i} - \mu_t p_t 
\quad\text{and}\quad \\
s_c = s_c(\sets) &= \frac1n{\sum_{i \in \mathcal{S}} Y_{ic}(1- p_i)} - \mu_c p_c. 
\end{split}
\end{align}
We will find that these quantities play an important role in bias.
If for instance the larger values of $\yit$ tend to co-occur with higher
propensities $p_i$ then averages are biased up.

The delta method variances of our estimators  depend on the following
weighted averages of squares and cross products 
\begin{equation}\label{eq:sttetc}
\begin{split}
S_{tt} = S_{tt}(\sets) &= \frac1n\sum_{i \in \sets} p_{i}(1- p_{i})(Y_{it} - \rho_{t})^2, \\
S_{cc} = S_{cc}(\sets) &= \frac1n\sum_{i \in \sets} p_{i}(1- p_{i}) (Y_{ic} - \rho_{c})^2,\quad \text{and}\\
S_{tc}= S_{tc}(\sets) &= \frac1n\sum_{i \in \sets}p_{i}(1- p_{i}) (Y_{it} - \rho_{t})(Y_{ic} - \rho_{c}),
\end{split}
\end{equation}
where $\rho_t=\rho_t(\sets)=\mu_t(\sets)+s_t(\sets)/p_t(\sets)$ 
and $\rho_c=\rho_c(\sets) = \mu_c(\sets)+s_c(\sets)/p_c(\sets)$. 
The quantity $\rho_t$ is the lead term in
$\edelt( \sum_{i\in\sets}W_{it}Y_{it}/\sum_{i\in\sets}W_{it})$
and $\rho_c$ is similar. More details about these quantities are in the Appendix
where Theorem~\ref{theorem:general} is proved.

\begin{proposition}\label{prop:scomp}
Let $\sets$ be $\odb_k$, $\rct_k$ or $\odb_k\cup\rct_k$.
Then under Assumption~\ref{assumption:strongtreatmenteffect}, 
$s_c(\sets)=-s_t(\sets)$.
\end{proposition}
\begin{proof}
Under Assumption~\ref{assumption:strongtreatmenteffect}, 
we can set $\yit=\yic+\tau_k$ and $\mu_t=\mu_c+\tau_k$ in~\eqref{eq:prcovfors}.
\end{proof}

\subsection{Main theorem}
We will compare the efficiency of our five estimators
using their delta method approximations.
We state two elementary propositions without proof
and then give our main theorem.  Results for our various
estimators are mostly direct corollaries of that theorem.

\begin{proposition}\label{prop:deltaratio}
Let $x$ and $y$ be jointly distributed random variables with 
means $x_{0}\ne0$ and $y_{0}$ respectively,
and finite variances. 
Let $\rho = y_0/x_0$. Then 
\begin{align}
\edelt\Bigl(\frac{y}{x}\Bigr) 
& =\rho -\frac{\cov(y-\rho x,x)}{x_0^2},\quad\text{and}\label{eq:edeltratio}\\[.5ex]
\vdelt\Bigl(\frac{y}{x}\Bigr) & = \frac{\var(y-\rho x)}{x_0^2}.\label{eq:vdeltratio}
\end{align}
\end{proposition}

\begin{proposition}\label{prop:deltaratiodiff}
Let $x_t$, $x_c$, $y_t$, $y_c$ be jointly distributed random variables with 
finite variances and means $x_{j,0}\ne0$ and $y_{j,0}$ respectively, for $j\in\{t,c\}$.
Let $\rho_j = y_{j,0}/x_{j,0}$. Then 
\begin{align*}
\vdelt\Bigl(
\frac{y_t}{x_t}\pm\frac{y_c}{x_c}
\Bigr) 
&=\frac{\var(y_t-\rho_t x_t)}{x_{t,0}^2}+\frac{\var(y_c-\rho_c x_c)}{x_{c,0}^2}
\pm2\frac{\cov(y_t-\rho_t x_t,y_c-\rho_c x_c)}{x_{t,0}x_{c,0}}.
\end{align*}
\end{proposition}

\begin{theorem}\label{theorem:general}
Let $\sets$ be an index set of finite cardinality $n>0$.
For $i\in \sets$, let $\wit\sim\bern(p_i)$ be independent
with $0<p_i<1$ and
set $\wic=1-\wit$.
Let
\[ \hat \tau = \frac{\sum_{i \in \mathcal{S}} \wit Y_{it}}{\sum_{i \in \mathcal{S}} \wit } - \frac{\sum_{i \in \mathcal{S}} \wic Y_{ic}}{\sum_{i \in \mathcal{S}} \wic }\] 
where $(Y_{it}, Y_{ic})\in[-B,B]^2$, for $B<\infty$. Then with
$\mu_t$, $\mu_c$, $p_t$, $p_c$, $s_t$, $s_c$, $S_{tt}$, $S_{cc}$, $S_{tc}$
defined at equations~\eqref{eq:pfors} through~\eqref{eq:sttetc},
\begin{align}
\vdelt(\hat \tau) &= \frac{1}{n}\left( \frac{S_{tt}}{p_t^2} + \frac{S_{cc}}{p_c^2} + 2 \frac{S_{tc}}{p_tp_c} \right) \label{eq:genvdelt}.
\end{align}
If all $p_i\in[\epsilon,1-\epsilon]$ for some $\epsilon>0$, then
\begin{align}\label{eq:genedelt}
\edelt\left( \hat \tau \right) 
&= \left( \mu_t - \mu_c \right) + \Bigl( \frac{s_{t}}{p_t} - \frac{s_c}{p_c} \Bigr) + O\Bigl( \frac{1}{n} \Bigr).
\end{align}
\end{theorem}
\begin{proof}
See Section~\ref{sec:proofthmgeneral}.
\end{proof}
The implied constant in $O(1/n)$ 
for equation~\eqref{eq:genedelt} holds for all $n\ge1$.

\subsection{Delta method means and variances}\label{sec:deltaMethodMeans}

We define the delta method bias of an estimate $\hat\tau_k$
via $\bdelt(\hat\tau_k)=\edelt(\hat\tau_k)-\tau_k$.

\begin{corollary}\label{cor:odbMoments}
Let $\hat\tau_{ok}$ be the ODB-only estimator from~\eqref{eq:deftauok}.
Then
\begin{align*}
\vdelt(\hat \tau_{ok}) &= 
\frac{1 }{n_{ok}}\left( \frac{S_{tt}}{ p_{t}^2} + \frac{S_{cc}}{p_{c}^2} + 2\frac{S_{tc}}{ p_{t}p_c} \right),
\end{align*}
where $s_t$, $s_c$, $p_t$, $p_c$, $S_{tt}$, $S_{cc}$ and $S_{cc}$ are 
given in equations~\eqref{eq:mufors} through~\eqref{eq:sttetc} with $\sets=\odb_k$. 
If $1<k<K$, then
\[ \bdelt\left( \hat \tau_{ok} \right) =  
\frac{s_{t}}{p_{t}}-\frac{s_{c}}{p_{c}}
+ O\biggl( \frac{1}{n_{ok}} \biggr). \]
If also Assumption~\ref{assumption:strongtreatmenteffect} holds, then
\[ \bdelt\left( \hat \tau_{ok} \right) =  
\frac{s_t}{p_t(1-p_t)}
+ O\biggl( \frac{1}{n_{ok}} \biggr). \]
\end{corollary}
\begin{proof}
For $1<k<K$ we can 
apply Theorem~\ref{theorem:general} with  $\epsilon = 1/K$.
Under Assumption~\ref{assumption:strongtreatmenteffect}, $s_c=-s_t$,
so the lead term in $\edelt(\hat\tau_k)$ is
 $s_t(1/p_t+1/p_c) =s_t(p_t+p_c)/p_t(1-p_t)
=s_t/p_t(1-p_t)$.
\end{proof}

\begin{corollary}\label{cor:rctMoments}
Let $\hat\tau_{rk}$ be the RCT-only estimator from~\eqref{eq:deftaurk}.
Then  
$$
\bdelt(\hat\tau_{rk}) = O\Bigl(\frac1{n_{rk}}\Bigr),
$$
and
\begin{align}\label{eq:vdeltrctonly}
\begin{split}
\vdelt( \hat \tau_{rk} ) &= 
\frac{\bar\sigma^2_{rk}}{n_{rk}p_r(1-p_r)},
\quad\text{where}\\
\bar\sigma^2_{rk} &= 
\frac1{n_{rk}}\sum_{i\in\rct_k} [ (\yit-\mu_{rkt})(1-p_r)+(\yic-\mu_{rkc})p_r]^2,
\end{split}
\end{align}
for $\mu_{rkt}=\mu_t(\rct_{k})$ and  $\mu_{rkc}=\mu_c(\rct_{k})$.
Under Assumption~\ref{assumption:strongtreatmenteffect},
$\bar\sigma^2_{rk} = \sigma^2_{rkt}\equiv (1/n_{rk})\sum_{i\in\rct_k}(\yit-\mu_{rkt})^2$.
If $p_r=1/2$, then
$$
\vdelt( \hat \tau_{rk} ) = \frac1{4n_k^2}\sum_{i\in\rct_k}\left( \bar Y_i - \frac{\mu_{rkt}+\mu_{rkc}}2\right)^2
$$
for $\bar Y_i = (\yit+\yic)/2$.
\end{corollary}
\begin{proof}
See Section~\ref{sec:proofcorrctMoments}.
\end{proof}

The RCT has a very tiny delta method bias which arises purely from the
ratio estimator (random denominator) form of $\hat\tau_{rk}$. Conditional on there being at least one treated and one control subject in the stratum, it can be shown that $\hat \tau_{rk}$ is exactly unbiased rather than just asymptotically unbiased. This follows from symmetry: at every value of $n_{rkt} \in \{1, 2, \dots, n_{rk}-1\}$, the estimator is drawn uniformly at random from all permutations of the labels of who is treated and who is not, and unbiasedness follows. 



\begin{corollary}\label{cor:weightedMoments}
Let $\hat\tau_{wk}$ be
the weighted-average estimator~\eqref{eq:deftauwk}.
Then, with $\lambda_k=n_{ok}/(n_{ok}+n_{rk})$,
\begin{align*}
\vdelt (\hat \tau_{wk}) 
&=\frac{\lambda_{k}}{n_{ok}+n_{rk}} \left( \frac{S_{tt}}{ p_{t}^2} + \frac{S_{cc}}{p_{c}^2} + 2\frac{S_{tc}}{ p_{t}p_c} \right) 
+\frac{1-\lambda_k}{n_{ok}+n_{rk}}
\frac{\bar\sigma^2_{rk}}{p_r(1-p_r)},
\end{align*}
where 
$S_{tt}$, $S_{cc}$ and $S_{cc}$ are 
given in equation~\eqref{eq:sttetc} with $\sets=\odb_k$,
and $\bar\sigma^2_{rk}$ is defined at~\eqref{eq:vdeltrctonly}. 
If $1<k<K$, then
\begin{align*}
\bdelt(\hat \tau_{wk}) &= 
\lambda_k 
\left( \frac{s_{okt}}{p_{okt}}-\frac{s_{okc}}{p_{okc}}\right)
+ O\left( \frac{1}{n_{ok} + n_{rk}} \right),
\end{align*}
where $s_{\okt}$, $p_{\okt}$, $s_{\okc}$ and $p_{\okc}$ are defined
by equations~\eqref{eq:pfors} and~\eqref{eq:prcovfors} for $\sets=\odb_{k}$.
If Assumption \ref{assumption:strongtreatmenteffect} also holds, then
\begin{align*}
\bdelt (\hat \tau_{wk}) &= 
 \frac{\lambda_k s_{okt}}{p_{okt}(1-p_{okt})}
+ O\left( \frac{1}{n_{ok} + n_{rk}} \right). 
\end{align*}
\end{corollary}
\begin{proof}
Using~\eqref{eq:edeltcombo} and Corollaries~\ref{cor:odbMoments} and~\ref{cor:rctMoments},
$\bdelt(\hat\tau_{wk})=\lambda_k\times\bdelt(\hat\tau_{ok})$ for $\lambda_k$ given in~\eqref{eq:deftauwk}.
This yields the lead terms in both expressions for $\bdelt(\hat\tau_{wk})$.
The error terms are $\lambda_kO(1/n_{ok})=O(1/(n_{ok}+n_{rk}))$.
Using independence of the RCT and ODB,
Corollaries~\ref{cor:odbMoments} and~\ref{cor:rctMoments}, 
and definition~\eqref{eq:vdeltcombo}
\begin{align*}
\vdelt (\hat \tau_{wk}) &=
\frac{\lambda_k^2}{n_{ok}} \left( \frac{S_{tt}}{ p_{t}^2} + \frac{S_{cc}}{p_{c}^2} + 2\frac{S_{tc}}{ p_{t}p_c} \right) 
+(1-\lambda_k)^2 \frac{\bar\sigma^2_{rk}}{n_{rk} p_r(1-p_r)}\\
&=\frac{\lambda_{k}}{n_{ok}+n_{rk}} \left( \frac{S_{tt}}{ p_{t}^2} + \frac{S_{cc}}{p_{c}^2} + 2\frac{S_{tc}}{ p_{t}p_c} \right) 
+\frac{1-\lambda_k}{n_{ok}+n_{rk}}
\frac{\bar\sigma^2_{rk}}{p_r(1-p_r)}.
\end{align*}
\end{proof}


In our motivating scenarios we anticipate that $n_o\gg n_r$ so that $\lambda_k\approx 1$ for most $k$.
Then the first term in $\vdelt(\hat\tau_{wk})$ is only slightly smaller than
$\vdelt(\hat\tau_{\ok})$ for the ODB-only estimate, and at most a small variance reduction is to be expected
from weighting.

The spiked-in estimator's bias and variance cannot be computed as a corollary of Theorem \ref{theorem:general}, but they can be computed directly. 

\begin{corollary}\label{cor:spikeMoments}
Let $\hat\tau_{sk}$ be the spiked-in estimator~\eqref{eq:deftauwk}. Then
\begin{align*}
\vdelt( \hat \tau_{sk}) &= 
\frac{1 }{n_{ok}+n_{rk}}\left( \frac{S_{tt}}{ p_{t}^2} + \frac{S_{cc}}{p_{c}^2} + 2\frac{S_{tc}}{ p_{t}p_c} \right),
\end{align*}
where $s_t$, $s_c$, $p_t$, $p_c$, $S_{tt}$, $S_{cc}$ and $S_{tc}$ are 
given in equations~\eqref{eq:mufors} through~\eqref{eq:sttetc} with $\sets=\odb_k\cup\rct_k$. 
If  $1<k<K$, then
\[ \bdelt\left( \hat \tau_{sk} \right) =  
\frac{s_{t}}{p_{t}}-\frac{s_{c}}{p_{c}}
+ O\biggl( \frac{1}{n_{ok}+n_{rk}} \biggr). \]
If Assumption~\ref{assumption:strongtreatmenteffect} also holds, then
\[ \bdelt\left( \hat \tau_{sk} \right) =  
\frac{s_t}{p_t(1-p_t)}
+ O\biggl( \frac{1}{n_{ok}+n_{rk}} \biggr). \]
\end{corollary}


\begin{proof}
The spike-in estimates are computed by pooling $\odb_k$ and $\rct_k$
into their union. 
\end{proof}

The edge bins are not covered by Corollary~\ref{cor:spikeMoments}.
Inspection of the proof of Theorem~\ref{theorem:general}
shows that the bias error term is $O( (n_{ok}+n_{rk})/n_{rk}^2)$.
No such bound is available for $\hat\tau_{ok}$ or $\hat\tau_{wk}$
for edge bins.


To relate the bias of $\hat\tau_{sk}$ to that of the other estimators, we write it in terms of the quantities computed using $\mathcal{S} = \mathcal{O}_k$ and $\mathcal{S} = \mathcal{R}_k$. Denoting these quantities using an additional subscript of $o$ and $r$, 
\begin{align}\label{eq:spikedBias}
\begin{split}
\bdelt\left( \hat \tau_{sk} \right) &=  \Delta_k n_{ok} \left( \frac{ p_{okt}}{n_{ok} p_{okt} + n_{rk} p_{rkt} } -  \frac{p_{okc}}{n_{ok} p_{okc} + n_{rk} p_{rkc} } \right) + \\ & \hspace{5mm} s_{okt}   \frac{ n_{ok}}{n_{ok} p_{okt} + n_{rk} p_{rkt} } - s_{okc}  \frac{ n_{ok}}{n_{ok} p_{okc} + n_{rk} p_{rkc} }
+ O\biggl( \frac{1}{n_{ok}+n_{rk}} \biggr).
\end{split}
\end{align}
The bias for $\hat\tau_{rk}$ is zero.  The bias for
 $\hat\tau_{ok}$ has terms analogous to the second and third (and error) terms above, but
the first term is new to $\hat\tau_{sk}$.
This term is linear in $\Delta_k$. For large values of $\Delta_k$, this term will dominate, yielding biases that can easily exceed those of $\hat \tau_{ok}$. This is the fundamental danger of the spiked-in estimator: if the mean potential outcomes differ substantially between ODB and RCT subjects with similar value of the propensity score function, then the estimation will be poor due to large bias.

\subsection{The dynamic weighted estimator}

The bias-variance tradeoffs are intrinsically different in each stratum.
Using results from the prior section, we derive a dynamic weighted estimator
that uses different weights in each stratum.
Our dynamic weighted estimator is based
on Assumption \ref{assumption:strongtreatmenteffect}, though
we will test it in settings where that assumption does not hold.

From Proposition~\ref{proposition:invMSE},
the MSE-optimal convex combination of $\hat\tau_{ok}$ and $\hat\tau_{rk}$ is
$c_{*k}\hat\tau_{ok}+(1-c_{*k})\hat\tau_{rk}$
where $c_{*k} = {\var(\hat\tau_{rk})}/({\var(\hat\tau_{rk})+\mse(\hat\tau_{ok})}).$
The dynamic weighted estimator is
\begin{align}\label{eq:deftaudk}
 \hat \tau_{dk} = \hat c_{*k} \hat \tau_{ok} + (1 - \hat c_{*k}) \hat \tau_{rk} ,
\quad\text{with}\quad
\hat c_{*k} = 
\frac{\wh\var(\hat\tau_{rk})}{\wh\var(\hat\tau_{rk})+\wh\mse(\hat\tau_{ok})},
\end{align}
for plug-in estimators of $\mse(\hat\tau_{ok})$
and $\var(\hat\tau_{rk})$. 
To obtain our MSE estimates we use
$\wt\mse(\cdot) = \bdelt(\cdot)^2+\vdelt(\cdot)$
taking the delta method moments from Corollaries~\ref{cor:odbMoments} and~\ref{cor:rctMoments}.
These expressions include some unknown population quantities
that we then approximate from the data to get $\wh\mse(\cdot)$.

For the ODB estimate we use
\begin{align*}
\wt\mse(\hat\tau_{ok})
&= 
\left(\frac{s_{t}}{p_{t}(1-p_{t})}\right)^2+
\frac{1 }{n_{ok}}\left( \frac{S_{tt}}{ p_{t}^2} + \frac{S_{cc}}{p_{c}^2} + 2\frac{S_{tc}}{ p_{t}p_c} \right)
\end{align*}
where the quantities on the right hand side are given in Section~\ref{sec:usefulQuantities}
with $\sets=\odb_k$.
For the RCT estimate we use
\begin{align*}
\wt\var(\hat\tau_{rk}) = \frac{\bar\sigma^2_{rk}}{p_r(1-p_r)n_{rk}},\quad\text{with}\quad
\bar\sigma^2_{rk} =  \frac1{n_{rk}}\sum_{i\in\rct_k}W_{it} \hat \sigma^2_{rkt} + W_{ic}\hat \sigma^2_{rkc}
\end{align*}
where $\hat \sigma^2_{rkt}, \hat \sigma^2_{rkc}$ are the sample variances observed among the treated and control units respectively. Both of these estimates use Assumption~\ref{assumption:strongtreatmenteffect}.

The values of $p_t$ and $p_c$ are known:
$p_t = \sum_{i\in\odb_k} p_{it}/n_{ok}$ where $p_{it}$ is the propensity $e(\bsx_i)$ and $p_c=1-p_t$. 
We use Horvitz-Thompson style inverse probability weighting to
estimate other quantities, as follows:
\begin{align*}
\hat \rho_t  &= \frac{\sum_{i \in \odb_k} W_{it} Y_{it}}{\sum_{i \in \odb_k} W_{it} }, \qquad \hat \rho_c  = \frac{\sum_{i \in \odb_k} W_{ic} Y_{ic}}{\sum_{i \in \odb_k} W_{ic} }, \\ 
\hat s_t &= \frac{\sum_{i \in \odb_k} W_{it}}{n_{ok}} \left(  \sum_{i\in\odb_k}\wit\yit - p_t \sum_{i\in\odb_k} W_{it} Y_{it}/p_{it}  \right)  \\ 
& \hspace{6mm} +\frac{\sum_{i \in \odb_k} W_{ic}}{n_{ok}} \left( \sum_{i\in\odb_k}\wic\yic - p_c \sum_{i\in\odb_k} W_{ic} Y_{ic}/p_{ic}\right), \\ 
\hat S_{tt} & = \frac{\sum_{i\in\odb_k}\wit p_{it} (1-p_{it})(\yit-\hat\rho_t)^2}{\sum_{i\in\odb_k}\wit},\quad\text{and}\\
\hat S_{cc} & = \frac{\sum_{i\in\odb_k}\wic p_{it} (1-p_{it})(\yic-\hat\rho_c)^2}{\sum_{i\in\odb_k}\wic}.
\end{align*}

The sole quantity that does not have a Horvitz-Thompson estimator is $S_{tc}(\odb_k)$, 
because we never observe both potential outcomes for a given unit. 
First, we write $S_{tc}$ as
\begin{align*}
\frac1n\sum_{i \in \odb_k}\wit p_{it}(1- p_{it}) (Y_{it} - \rho_{t})(Y_{ic} - \rho_{c}) 
+\frac1n\sum_{i \in \odb_k}\wic p_{it}(1- p_{it}) (Y_{it} - \rho_{t})(Y_{ic} - \rho_{c}).
\end{align*}
Next, under Assumption~\ref{assumption:strongtreatmenteffect},
$$\yit-\rho_t = \yic+\tau_k-\mu_t-s_t/p_t 
= \yic - \rho_c - \frac{s_t}{p_t p_c},$$
and similarly $\yic-\rho_c=
\yit - \rho_t + {s_t}/(p_t p_c)$.
Therefore
\begin{align}\label{eq:stchat}
\begin{split}
S_{tc}&= \frac1n\sum_{i \in \odb_k}\wit p_{it}(1- p_{it}) (Y_{it} - \rho_{t})^2+\frac1n\sum_{i \in \odb_k}\wic p_{it}(1- p_{it}) (\yic-\rho_c)^2 - \\
&\phantom{=} - \frac{s_{t}}{n p_{t}(1-p_t)}\left( \sum_{i \in \mathcal{O}_k}  W_{it} p_{it}(1- p_{it}) (Y_{it} - \rho_{t}) - W_{ic} p_{it}(1- p_{it}) (Y_{ic} - \rho_{c}) \right) 
\end{split}
\end{align}
and we get $\hat S_{tc}$ by plugging the above estimates
of $\rho_t$, $\rho_c$ and known values of $p_t$, $p_c$
into~\eqref{eq:stchat}.
Although Assumption~\ref{assumption:strongtreatmenteffect} is used to derive the
estimator, some of our simulations in Section~\ref{sec:simulations}  test it under a violation of that assumption.

\subsection{Performance comparison}\label{sec:perfcomparison}

The ideal dynamic estimator with the optimal weight $c_{k*}$ must be at
least as good as $\hat\tau_{ok}$, $\hat\tau_{rk}$ and $\hat\tau_{wk}$ because
those estimators are all special cases of weighting estimators belonging to
the class that $c_{k*}$ optimizes over.  Our estimator $\hat\tau_{dk}$ will not
always be better than those other estimators, because it uses an estimate $\hat c_{k*}$
which could introduce enough error to make it less efficient. 

When combining stratum-based estimates $\hat \tau_k$
into the weighted estimator $\hat\tau = \sum_k\omega_k\hat\tau_k$,
there is the possibility of biases canceling between strata.
None of the competing estimators we consider are designed to exploit
such cancellation.
For large strata, $c_{k*}$ should be well estimated. To arrange cancellations
among biased within-stratum estimates would require domain-specific
assumptions that we do not make here.

The comparison to the spiked-in estimator is more complex. As we saw in equation~\eqref{eq:spikedBias}, the bias can grow without bound in $\Delta_k$, so for large $\Delta_k$ this estimator will have the largest MSE. However, for small values of $\Delta_k$, the spiked-in estimator can outperform all the other estimators. To see why, we make a direct comparison with the dynamic weighted estimator and reference our prior discussion showing the dynamic weighted estimator will generally outperform $\hat \tau_{ok}$, $\hat \tau_{rk}$ and $\hat \tau_{wk}$.

We introduce sample counterparts of $\Delta_k$, given by
\begin{align*}
\hat \Delta_{kt}
&=
\frac{\sum_{i \in \mathcal{O}_k} W_{it} Y_{it} }{\sum_{i \in \mathcal{O}_k} W_{it}} -\frac{\sum_{i \in \mathcal{R}_k} W_{it} Y_{it} }{\sum_{i \in \mathcal{R}_k} W_{it}},\quad\text{and}\\
\hat \Delta_{kc} & = 
\frac{\sum_{i \in \mathcal{O}_k} W_{ic} Y_{ic} }{\sum_{i \in \mathcal{O}_k} W_{ic}} -\frac{\sum_{i \in \mathcal{R}_k} W_{ic} Y_{ic} }{\sum_{i \in \mathcal{R}_k} W_{ic}}. 
\end{align*}
Then after some algebra $\hat\tau_{sk}$ differs from the RCT estimate
as follows,
\begin{align}\label{eq:spikeminusrct}
\hat\tau_{sk} -\hat\tau_{rk} =c_{kt} \hat\Delta_{kt} - c_{kc} \hat\Delta_{kc}
\end{align}
for sample size proportions
\begin{align*}
c_{kt} = \frac{\sum_{i \in \odb_k} W_{it}}{\sum_{i \in \odb_k\cup \rct_k} W_{it} }\quad\text{and}\quad c_{kc} = 
\frac{\sum_{i \in \odb_k} W_{ic}}{\sum_{i \in \odb_k\cup\rct_k} W_{ic} }.
\end{align*}
By comparison, 
\begin{align}\label{eq:dynamicminusrct}
\hat\tau_{dk}-\hat\tau_{rk}=
 c_{k\star} \hat\Delta_{kt} - c_{k\star} \hat\Delta_{kc},
\end{align}
where the dynamic estimator tunes $c_{k\star}$ to the available
data.  An oracle could choose $c_{k\star}$ optimally 
using Proposition~\ref{proposition:invMSE}.
While the oracle is working in a one parameter family~\eqref{eq:dynamicminusrct}
 for each bin $k$,  the
spiked-in estimator uses two weights $c_{kt}$ and $c_{kc}$~\eqref{eq:spikeminusrct}
that are not necessarily within the family that the oracle optimizes over.
This is why it is possible for the spiked estimator to outperform the oracle.




\section{Simulations}\label{sec:simulations}

Our goal is to estimate the average treatment effect in the target population, from which we assume the ODB data was randomly sampled. 
The value of the RCT is that it can substitute for ODB data in places where that data is sparse due to the treatment assignment mechanism. 

We simulate two high level scenarios.  In one, the RCT is a random sample from the same population that the ODB came from.  Then the RCT and ODB data differ only in their treatment assignment mechanisms. We consider this case the ideal one for our approach of merging the RCT into the ODB.
In the other scenario, the RCT is subject to some potentially biasing inclusion criteria on the explanatory
variables $\bsx_i$.
Such biases are a frequent concern for RCTs \citep{susukida2016assessing, stuart2017generalizing}. 

For both of these high level scenarios, we vary the treatment effect over strata, making it either
constant, linear or quadratic in $k=1,\dots,K$.
Section~\ref{sec:ideal} shows results for our ideal case
where $\bsx_i$ have the same distribution in both data sets and
Assumption~\ref{assumption:strongtreatmenteffect} holds.
Section~\ref{sec:restrict} models a sampling bias for the $\bsx_i$ values  in the RCT
while retaining Assumption~\ref{assumption:strongtreatmenteffect}.
Section~\ref{sec:withouta4} removes Assumption~\ref{assumption:strongtreatmenteffect}
from both of the prior simulation settings.

\subsection{Simulation of the ideal case}\label{sec:ideal}

We begin with the simulations satisfying Assumption~\ref{assumption:strongtreatmenteffect},
with the RCT sampled from the same distribution as the ODB. This is an ideal case. 
First we describe how the ODB data are generated, then the RCT data.

In all of our simulations $\bsx_i\in\real^5$.
The ODB has $n_o=5{,}000$ subjects. 
On each new sampling of the covariates, we first sample a covariance matrix $\Sigma \in \mathbb{R}^5$, such that each covariate has unit variance, and covariances are randomly 0 with $1/2$ probability, and $\pm 0.1$ with $1/4$ probability. Such a covariance structure is, on average, roughly consistent with the covariance structure in the real dataset in Section \ref{sec:whi}. We then generate $\bsx_i\simiid\dnorm(0,\Sigma)$ for $i\in\odb$ and we assume that for the control condition:
$$
\yic = \bsx_i^\tran\beta + \err_i,\quad\text{for $\beta = (1,1,1,1,1)^\tran$}
$$
where $\err_i$ are generated as IID $\dnorm(0,1)$ random variables.
We assume that the user does not know the precise form of the generative model
and uses the stratified estimates we presented above.

The treatment variables  in the ODB are independent Bernoulli random variables with
\[\Pr(W_i = 1) = \frac{1}{1 + e^{-\gamma^\tran\bsx_i}}.\] 
We consider the four $\gamma$ vectors given in Table~\ref{tab:gammas}.
Two of them are orthogonal to $\beta$. The others are correlated with $\beta$
and will result in the test group having higher average responses in the ODB
than the control group. For each correlation pattern we have two sizes
of $\Vert\gamma\Vert$.

\begin{table}\centering 
\begin{tabular}{ccccc}
$\gamma$: 
& $\begin{pmatrix} 1\\ 1\\ 1\\ 0\\ 0\end{pmatrix}$
& $\begin{pmatrix} \sqrt{2}\\ \sqrt{2}\\ \sqrt{2}\\ 0\\ 0\end{pmatrix}$
& $\begin{pmatrix}\phm\sqrt{3}/2 \\ -\sqrt{3}/2\\ \phm\sqrt{3}/2 \\ -\sqrt{3}/2 \\\phm 0 \end{pmatrix}$
& $\begin{pmatrix}\phm\sqrt{6}/2 \\ -\sqrt{6}/2\\ \phm\sqrt{6}/2 \\ -\sqrt{6}/2 \\\phm 0 \end{pmatrix}$\\
\\
$\gamma^\tran\beta$: & $3$ & $3\sqrt{2}$ & $0$ & $0$ \\
$\Vert\gamma\Vert^2$: & $3$ & $6$ & $3$ & $6$ \\
\end{tabular}
\caption{\label{tab:gammas}
These are the four $\gamma$ vectors used in our simulations. 
The first two correlate with the mean response vector $\beta$, while the 
second two do not.  The second and fourth imply larger sampling biases 
than the first and third do. 
}
\end{table}

The treatment values $\yit$ are equal to $\yic$ plus a treatment effect
that obeys Assumption~\ref{assumption:strongtreatmenteffect}. We use three
structures.  In all cases $\yit=\yic + \tau_k$ for $i\in\odb_k$.
The values of $\tau_k$ in constant, linear and quadratic treatment effect models are
\begin{align}\label{eq:treatments}
\tau_k  = T, \quad \tau_k = T\times \frac{k}{K},\quad\text{and}\quad
\tau_k  = T\times\Bigl(\frac{k}{K} -\frac12\Bigr)^2
\end{align}
respectively. In each case we choose the scale $T>0$
so that Cohen's $d$ \citep{cohen1988statistical} in the ODB precisely equals $0.5$, which Cohen calls a medium effect size. The value of $T$ varies across simulations based on the simulation settings and the randomness in the sampling of $\bsx_i$.  


We turn now to the RCT. It has $n_r=200$ subjects in it.
Their $\bsx_i$ are chosen from the same distribution as in the ODB (including the same covariance matrix $\Sigma$), 
but now $W_i\simiid\bern(1/2)$.
The same constant, linear and quadratic treatment effects from the ODB are used in the RCT.

We simulate the covariates $\bsx_i$ and potential outcomes $(\yit,\yic)$ for $i\in\odb\cup\rct$  
$100$ times.  For each realization of the covariates and potential outcomes
we make $20$ independent simulations of the treatment variables $\wi$, 
for a total of $2{,}000$ simulations.
In all cases we choose $K=20$ propensity bins where the $k$'th one
has $e(\bsx_i)\in[(k-1)/20,k/20)$ for $k=1,\dots,K$.
This simulation satisfies Assumptions~\ref{assumption:odb}
through~\ref{assumption:strongtreatmenteffect}.

In each simulation run, we estimate the average treatment effect using each of our five estimators. We also estimate using an `oracle' estimator, which knows the true MSE of the ODB-only and RCT-only estimators in each stratum and thus the optimal weighting between these estimators. The MSE of the estimators are computed across all 2,000 simulation runs,. 

The results are shown in Table~\ref{tab:resuideal}.
In this ideal setting, the spike-in estimator always has the lowest MSE. It was always
superior to the better of the RCT and ODB estimators.  Even though the RCT only adds $200$
subjects to the $5{,}000$ in the ODB, it leads to a spiked-in estimator whose MSE
ranges from about 50\% to about 80\% of that of the ODB.
As we discussed in Section \ref{sec:perfcomparison}, this setting has $\Delta_k=0$
and the RCT and ODB points have similar variances conditionally on the propensity score.
It even outperforms
the oracle estimator which optimizes the relative weights on the RCT and ODB within strata.
The spiked-in estimator can beat the oracle because it is not
one of the weighting schemes over which the oracle has optimized.
The oracle is the second-best performer in each condition. The dynamic weighted estimator -- which seeks to recover the oracle weights -- has an MSE only slightly inflated relative to the oracle, with no performance gap larger than 15\%. The dynamic weighted estimator also generally outperforms the weighted-average estimator. 

\begin{table}[h]
\centering 
\begin{tabular}{cccccccccc}
\toprule 
 \textbf{Trt.} & \textbf{Cor.} &$\boldsymbol{\Vert\gamma\Vert_2^2}$ & \textbf{ODB}    & \textbf{RCT}    & \textbf{Wtd.} & \textbf{Spike}   & \textbf{Dyn.} & \textbf{Oracle} \\ 
\midrule 
c&y&3      & 0.0069 & 0.0776 & 0.0066 & 0.0053 & 0.0065 & 0.0058 \\
c&y&6      & 0.0222 & 0.0780 & 0.0207 & 0.0113 & 0.0134 & 0.0129 \\
c&n&3      & 0.0117 & 0.1522 & 0.0110 & 0.0091 & 0.0110 & 0.0099 \\
c&n&6      & 0.0209 & 0.1457 & 0.0198 & 0.0138 & 0.0182 & 0.0163 \\
\midrule 
l&y&3      & 0.0076 & 0.0761 & 0.0071 & 0.0056 & 0.0066 & 0.0060 \\
l&y&6      & 0.0220 & 0.0787 & 0.0204 & 0.0111 & 0.0132 & 0.0128 \\
l&n&3      & 0.0122 & 0.1574 & 0.0116 & 0.0094 & 0.0118 & 0.0104 \\
l&n&6      & 0.0219 & 0.1434 & 0.0204 & 0.0137 & 0.0177 & 0.0160 \\
\midrule 
q&y&3     & 0.0077 & 0.0766 & 0.0072 & 0.0054 & 0.0066 & 0.0060 \\
q&y&6     & 0.0236 & 0.0791 & 0.0220 & 0.0113 & 0.0143 & 0.0139 \\
q&n&3     & 0.0122 & 0.1536 & 0.0115 & 0.0096 & 0.0116 & 0.0103 \\
q&n&6     & 0.0202 & 0.1414 & 0.0189 & 0.0124 & 0.0167 & 0.0152 \\
\bottomrule 
\end{tabular}
\caption{\label{tab:resuideal}
MSEs for treatment effect in the ideal setting.  Column 1 gives treatment (constant, linear, quadratic). Column 2 shows  whether the 
propensity was correlated with the mean response. Column 3 indicates the magnitude of the propensity vector $\gamma$. The remaining 
columns are mean squared errors for the overall treatment from our 5 estimators and an oracle.  In every case, the spiked-in estimator 
using~\eqref{eq:deftausk} has lowest MSE.}
\end{table}

\begin{figure}[H]
\centering
\textbf{Estimator Performance: Quadratic Treatment Effect, Ideal Case}\par\medskip
\includegraphics[scale = 0.4]{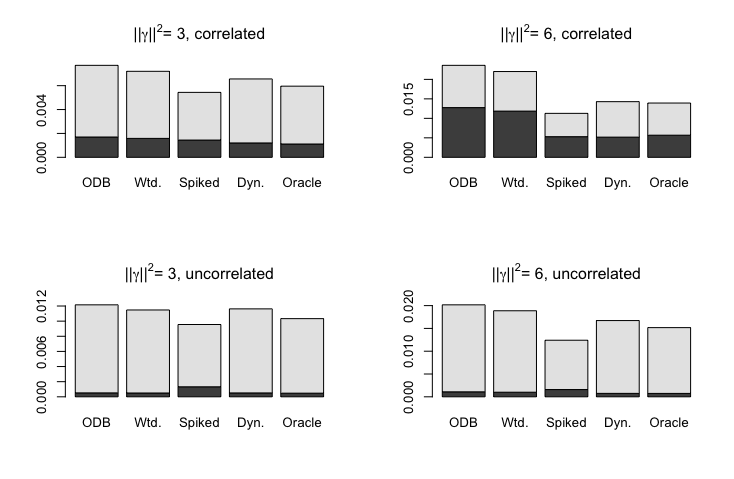}
\caption{\label{fig:quadSRS} Performance measures across all 2,000 simulations run in the ideal case. Bias squared is shown in black, and variance in gray, so that total bar height represents the MSE. The much larger values for the RCT estimator are excluded to make visual comparison easier. }
\end{figure}

The outcomes in Table~\ref{tab:resuideal} are quite consistent.
Nine out of  $12$ settings have the same ordering. 
From best to worst they are: spiked, oracle, dynamic, weighted, ODB and RCT.
In two of the cases (rows 3, 7, and 11) the weighted method very slightly outperforms the dynamic method. 

Inspection of the data in Table~\ref{tab:resuideal} shows that the RCT-only estimate is far
from competitive. This is not surprising as that estimate uses 
much less data than the other methods.
To make comparison easier, we exclude the RCT estimator
from our graphical presentation in Figure \ref{fig:quadSRS}. Bias squared is shown in black, and variance in gray, so that total bar height represents the MSE. The treatment effect patterns make very little difference, so we show only the case of the quadratic treatment effect. 
These results show that the advantage of the spiked-in estimator is greatest when $\Vert\gamma\Vert$ is large. In this case, the benefit mostly accrues due to a reduced variance for the spiked-in estimator relative to the dynamic estimator. 


\subsection{Restrictive enrollment criteria}\label{sec:restrict}

It is common for an RCT to have enrollment criteria such that the values of $\bsx_i$ in it are different
from those in the general population.   The RCT might be designed to avoid frail patients. Or it might be designed
to include patients with the worst prognoses, who are most in need of a novel treatment.  
We illustrate restrictive enrollment by having the RCT sample $\bsx_i$ from $\dnorm(0,\Sigma)$ subject to both
$x_{i1} < -1$ and $x_{i5}<-1$. Because our $\beta$ vector has all positive entries, 
these  restrictions mean that subjects in the RCT tend
to have smaller values of $\yic$ than those in the ODB.  Smaller could either mean better or
worse depending on what quantity $Y$ measures.
The first restriction is on a variable that influences the propensity for treatment in the ODB,
while the second restriction is on a variable that does not influence this propensity.

\begin{table}
\centering
\begin{tabular}{ccccccccc}
\toprule
 \textbf{Trt.} & \textbf{Cor.} &$\boldsymbol{\Vert\gamma\Vert_2^2}$ & \textbf{ODB}    & \textbf{RCT}    & \textbf{Wtd.} & \textbf{Spike}   & \textbf{Dyn.} & \textbf{Oracle} \\ 
\midrule
c&y&3           & 0.0075 & 0.1505 & 0.0068 & 0.0119 & 0.0060 & 0.0056 \\
c&y&6           & 0.0222 & 0.1681 & 0.0200 & 0.0186 & 0.0129 & 0.0121 \\
c&n&3          & 0.0114 & 0.1623 & 0.0106 & 0.0320 & 0.0104 & 0.0093 \\
c&n&6          & 0.0212 & 0.1965 & 0.0192 & 0.0667 & 0.0158 & 0.0140 \\
\midrule 
l&y&3           & 0.0074 & 0.4137 & 0.0068 & 0.0129 & 0.0062 & 0.0058 \\
l&y&6           & 0.0226 & 0.4918 & 0.0203 & 0.0191 & 0.0132 & 0.0123 \\
l&n&3          & 0.0124 & 0.2578 & 0.0115 & 0.0358 & 0.0109 & 0.0098 \\
l&n&6          & 0.0207 & 0.3015 & 0.0189 & 0.0627 & 0.0161 & 0.0141 \\
\midrule 
q&y&3         & 0.0075 & 0.3854 & 0.0069 & 0.0119 & 0.0062 & 0.0057 \\
q&y&6         & 0.0222 & 0.2886 & 0.0201 & 0.0186 & 0.0131 & 0.0124 \\
q&n&3        & 0.0126 & 0.2984 & 0.0116 & 0.0356 & 0.0109 & 0.0101 \\
q&n&6       & 0.0214 & 0.2569 & 0.0196 & 0.0684 & 0.0169 & 0.0150 \\
\bottomrule
\end{tabular}
\caption{\label{tab:resurestricted}
MSEs for treatment effect in the setting with restricted enrollments.
The columns are the same as in Table~\ref{tab:resuideal}.
Here the oracle estimator is always best and the dynamic estimator
is the best of the ones that can be implemented.
}
\end{table}

\begin{figure}  
\centering
\textbf{Estimator Performance: Quadratic Treatment Effect, Restricted Enrollment Case}\par\medskip
\includegraphics[scale = 0.4]{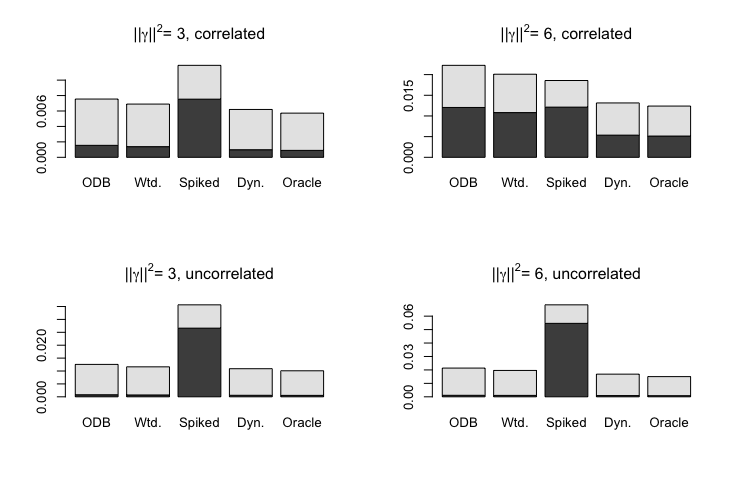}
\caption{\label{fig:quadRestrict} Performance measures across all 2,000 simulations run in the restricted enrollment case. Bias squared is shown in black, and variance in gray, so that total bar height represents the MSE. The much larger values for the RCT estimator are excluded to make visual comparison easier.}
\end{figure}

The results under this restriction are presented in Table~\ref{tab:resurestricted}. 
Nine of the $12$ settings have the same ordering. 
From best to worst they are:
oracle, dynamic, weighted, ODB, spiked, and RCT. In the remaining three cases, the spiked estimator is third, the weighted estimator fourth, and the ODB estimator fifth, with the other rankings unchanged. Across all settings, the best usable method is the dynamic one.  The dynamic MSE  was never more than 
$15$\% higher than that of the oracle that it seeks to approximate.  
Sometimes it was up to 40\% more efficient than the ODB-only estimator. 

The bad performance of the spiked-in estimator here is attributable to the restriction
on the fifth component of $\bsx_i$.  That restriction affects the outcomes $(\yic,\yit)$
but not the propensity score because $\gamma_5=0$. The result is a large $\Delta_k$
in each stratum, making the spiked-in estimator perform much worse than in the prior scenario. This effect can be seen in Figure \ref{fig:quadRestrict}, where the spiked-in estimator consistently demonstrates a large squared bias, resulting in a high MSE.



\subsection{Violation of Assumption~\ref{assumption:strongtreatmenteffect}}\label{sec:withouta4}

In this section we simulate in a setting where Assumption~\ref{assumption:strongtreatmenteffect} fails to hold
but Assumption~\ref{assumption:treatmenteffect} does hold.
We modify the linear and quadratic treatment effects in equation~\eqref{eq:treatments} to have
$$
\yit - \yic = T\times e(\bsx_i),\quad\text{and}\quad
\yit - \yic = T\times(e(\bsx_i)-1/2)^2
$$
respectively. $T$ is again selected so that Cohen's $d$ in the ODB is equal to 0.5. 
The treatment difference now depends on the actual propensity of each subject but it
varies with the strata.  We do not re-simulate the constant case because
it is the same either way.


If the RCT is sampled randomly from the population, we get the results in 
Table~\ref{tab:almostideal}.
Here again, the rankings are very stable. Seven times out of $8$, the ranking
from best to worst is: spiked, oracle, dynamic, weighted, ODB and RCT.
One time the weighted estimator slightly outperformed the dynamic estimator. 
The orderings are essentially unchanged from the case when Assumption~\ref{assumption:strongtreatmenteffect} held.
The value of $\Delta_k$ here, while not zero, is not very large.
The dynamic estimator is still a good approximation to the oracle, with an MSE never more than $14$\% larger.

\begin{table}
\centering
\begin{tabular}{ccccccccc}
\toprule
 \textbf{Trt.} & \textbf{Cor.} &$\boldsymbol{\Vert\gamma\Vert_2^2}$ & \textbf{ODB}    & \textbf{RCT}    & \textbf{Wtd.} & \textbf{Spike}   & \textbf{Dyn.} & \textbf{Oracle} \\ 
\midrule
l&y&3 & 0.0077 & 0.0812 & 0.0073 & 0.0055 & 0.0068 & 0.0061 \\
l&y&6 & 0.0243 & 0.0787 & 0.0226 & 0.0113 & 0.0148 & 0.0144 \\
l&n&3 & 0.0116 & 0.1456 & 0.0110 & 0.0091 & 0.0113 & 0.0100 \\
l&n&6 & 0.0210 & 0.1443 & 0.0197 & 0.0137 & 0.0178 & 0.0157 \\
\midrule 
q&y&3 & 0.0070 & 0.0806 & 0.0066 & 0.0050 & 0.0062 & 0.0056 \\
q&y&6 & 0.0191 & 0.0772 & 0.0177 & 0.0091 & 0.012 & 0.0117 \\
q&n&3 & 0.0122 & 0.152 & 0.0116 & 0.0092 & 0.0113 & 0.0102 \\
q&n&6 & 0.0209 & 0.1575 & 0.0195 & 0.0136 & 0.0179 & 0.0160 \\
\bottomrule
\end{tabular}
\caption{\label{tab:almostideal}
These are the results of the simulations where Assumption~\ref{assumption:strongtreatmenteffect} is violated 
but the RCT has the same $\bsx$ distribution as the ODB. 
}
\end{table}

Finally, we consider the setting where  
Assumption~\ref{assumption:strongtreatmenteffect} is violated 
and the RCT has the enrollment restrictions from Section~\ref{sec:restrict}.
The results are in Table~\ref{tab:almostnotideal}.
In $6$ of $8$ cases, the ranking is: oracle, dynamic, weighted, ODB, spiked and RCT, just as it predominantly was when Assumption~\ref{assumption:strongtreatmenteffect}
held. The two dissimilar cases still have the oracle and dynamic estimators as the top performers. Taken together, these results indicate that the dynamic weighted estimator is robust to this type of weakening of
Assumption~\ref{assumption:strongtreatmenteffect}.

\begin{table}
\centering
\begin{tabular}{cccccccccc}
\toprule
 \textbf{Trt.} & \textbf{Cor.} &$\boldsymbol{\Vert\gamma\Vert_2^2}$ & \textbf{ODB}    & \textbf{RCT}    & \textbf{Wtd.} & \textbf{Spike}   & \textbf{Dyn.} & \textbf{Oracle} \\ 
\midrule
l&y&3&       0.0077 & 0.4037 & 0.0071 & 0.0128 & 0.0064 & 0.0059 \\
l&y&6&       0.0247 & 0.5027 & 0.0223 & 0.0186 & 0.0142 & 0.0133 \\
l&n&3&       0.0127 & 0.2723 & 0.0117 & 0.0328 & 0.0111 & 0.0101 \\
l&n&6&       0.0201 & 0.3293 & 0.0185 & 0.0657 & 0.0155 & 0.0142 \\
\midrule 
q&y&3&      0.0068 & 0.2771 & 0.0063 & 0.0134 & 0.0058 & 0.0054 \\
q&y&6&      0.0211 & 0.2451 & 0.0194 & 0.0173 & 0.0137 & 0.0130 \\
q&n&3&      0.0126 & 0.2547 & 0.0116 & 0.0356 & 0.0108 & 0.0100 \\
q&n&6&      0.0216 & 0.2137 & 0.0199 & 0.0683 & 0.0168 & 0.0152 \\
\bottomrule
\end{tabular}
\caption{\label{tab:almostnotideal}
These are the results of the simulations where Assumption~\ref{assumption:strongtreatmenteffect} is violated 
and the $\bsx$ in the  RCT are subject to restrictive enrollment criteria.
}
\end{table}

Additional simulations were also conducted to explore the effect of the covariance structure between the covariates. Results were substantively similar with both $\Sigma = I_5$ -- i.e., independent covariates -- as well as with stronger correlations among the covariates. Performance of the spiked-in estimator was seen to degrade somewhat less with restrictive enrollment criteria when covariance was high, and somewhat more in the independent case. These simulations are omitted here for brevity. 

\section{WHI data example}\label{sec:whi}



In this section we evaluate our estimators on data from the Women's Health Initiative (WHI)
to estimate the effect of hormone therapy (HT) on coronary heart disease (CHD).
The WHI is a study of postmenopausal women in the United States, consisting of randomized controlled trial and observational study components with 161,808 total women enrolled \citep{prentice2005combined}. Eligibility and recruitment data for the WHI can be found in \cite{hays2003women} and \cite{writing2002risks}. Participants were women between 50 and 79 years old at baseline, who had a predicted survival of at least three years and were unlikely to leave their current geographic area for three years. 

Women with a uterus who met various safety, adherence, and retention criteria were eligible for a combined hormone therapy trial. A total of 16,608 women were randomized to the treatment arm, 8,506 women of whom were were assigned to take 625 mg of estrogen, and 2.5 mg of progestin, and the remainder of whom received a placebo. A corresponding 53,054 women in the observational component of the WHI had an intact uterus and were not using unopposed estrogen at baseline, thus rendering them comparable according to \cite{prentice2005combined} . About a third of these women were using estrogen plus progestin, while the remaining women in the observational study were not using hormone therapy \citep{prentice2005combined}. 

Participants received semiannual contacts and annual in-clinic visits for the collection of information about outcomes. 
Disease events, including CHD, were first self-reported and later adjudicated by physicians. We focus on outcomes during the initial phase of the study, which extended for an average of 8.16 years of follow-up in the RCT and 7.96 years in the ODB.


The overall rate of CHD in the trial was 3.7\% in the treated group (314 cases among 8,472 women) versus 3.3\% (269 cases among 8,065 women) among women not randomized to estrogen and progestin. In the observational study, the corresponding rates were 1.6\% among treated women (706 out of 17,457 women) and 3.1\% among control women (1,108 out of 35,408 women). Our methodology compares means and not survival curves. In the initial follow-up period, death rates were relatively low in both the ODB (6.4\%) and the RCT (5.7\%). Hence, we do not make corrections for the possibility of these deaths censoring CHD events.


\subsection{Covariate imbalance and modeling}

The WHI researchers collected a rich set of covariates about the participants in the study. For the purposes of computational speed, we narrow to a set of 684 variables, spanning demographics, medical history, diet, physical measurements, and psychosocial data collected at baseline. 

A crude measure of covariate imbalance can be found by looking at the distribution of p-values associated with each covariate. For continuous covariates, the p-value is computed by regressing the treatment indicator on the covariate and computing the p-value of the coefficient; for categorical covariates, it is computed via an a $\chi^2$ table test. 

If the selection into the treated and control groups is dependent on the covariates, then we would expect a p-value distribution highly skewed towards 0. If the selection is independent of the covariates (as we expect for an RCT), the p-value distribution should be approximately uniform. This is precisely what we see in Figure \ref{fig:Biases}. 

\begin{figure}[h]
\caption{\label{fig:Biases} Imbalance p-value distribution across 684 covariates in the observational (OS) and estrogen and progestin randomized controlled trial (E+P RCT) data. }
\centering
\includegraphics[width = 0.49\textwidth]{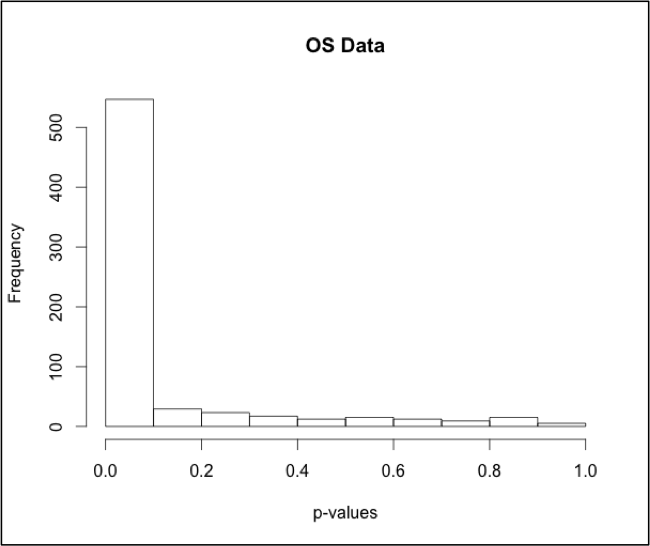}
\includegraphics[width = 0.49\textwidth]{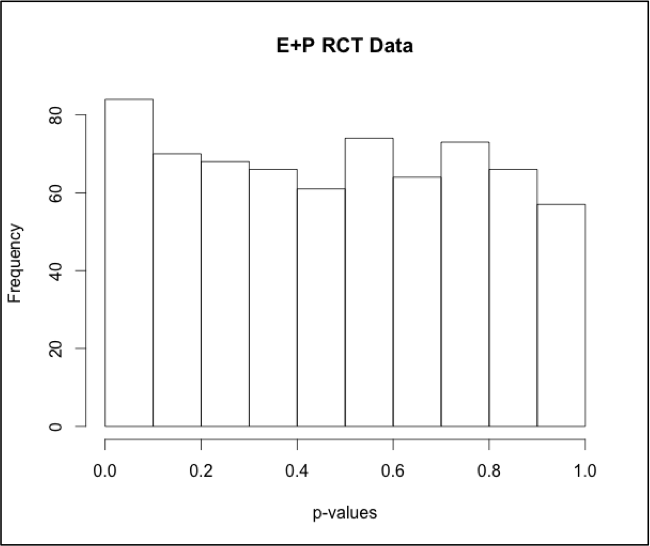}
\end{figure}

A more meaningful measure of covariate imbalance can be found by looking at clinically relevant factors. 
\cite{prentice2005combined} identified factors that are correlated with CHD. They found that 
HT users in the observational study were more likely to be Caucasian or Asian/Pacific Islander, less likely to be overweight, and more likely to have a college degree. These imbalances strongly suggest that applying a naive differencing estimate to the observational data will yield an unfairly rosy view of the effect of hormone therapy (HT) on CHD. 


To generate our estimators for this dataset, we need a propensity model $e(\bsx)$ to map the observed covariates to an estimated probability of receiving the treatment in the observational study. In constructing this model, we have two potentially conflicting goals:
\begin{compactenum}[\quad\ 1)]
\item The model should fully account for the effect of covariates on the selection into the treatment group  in the observational study, and 
\item the model should generalize to the RCT such that we an obtain an accurate estimate of what probability of receiving HT \emph{would have been} for RCT participants had they not been randomized to treatment or control. 
\end{compactenum}
We do not need $e(\bsx_i)$ to have a causal interpretation for the treatment. It only needs to have a strong association with the assignment of subject $i$ to the test or control treatment.

To achieve these goals, we used a logistic regression-based procedure designed to generate an expressive model while limiting overfit. We also explored more complex models, such as random forests and generalized boosted trees, but found little benefit along with a dramatic increase in runtime. 

A forward stepping algorithm was first applied to the observational dataset to put an ordering on the variables. All 684 baseline covariates were provided as candidates to a logistic regression predicting the treatment indicator, and variables were automatically added, one at a time, based on which addition most reduced Akaike's Information Criterion \citep{akaike1974new}. 


Using this ordering, models containing from one to 120 variables were generated. Model fit was assessed via the area under the Receiver Operator Characteristic curve, or ROC AUC. At each model size, the ROC AUC was computed first for the nominal model and then computed again using a ten-fold cross-validation. This procedure generated the curves seen in Figure \ref{fig:cvROCAUC}. Notably, we observe that the predictive power rises rapidly with the addition of the first twenty variables to the logistic regression model, but slows dramatically thereafter. There is also very little evidence of overfit, as the nominal AUC only very slightly outpaces the cross-validated AUC, even in models with 100 or more variables. This is likely a consequence of the sheer number of observations in the OS dataset. 

\begin{figure}[h]
\caption{\label{fig:cvROCAUC} Nominal and cross-validated receiver operator characteristic area under curve for propensity models with different numbers of variables}
\centering
\includegraphics[width = \textwidth]{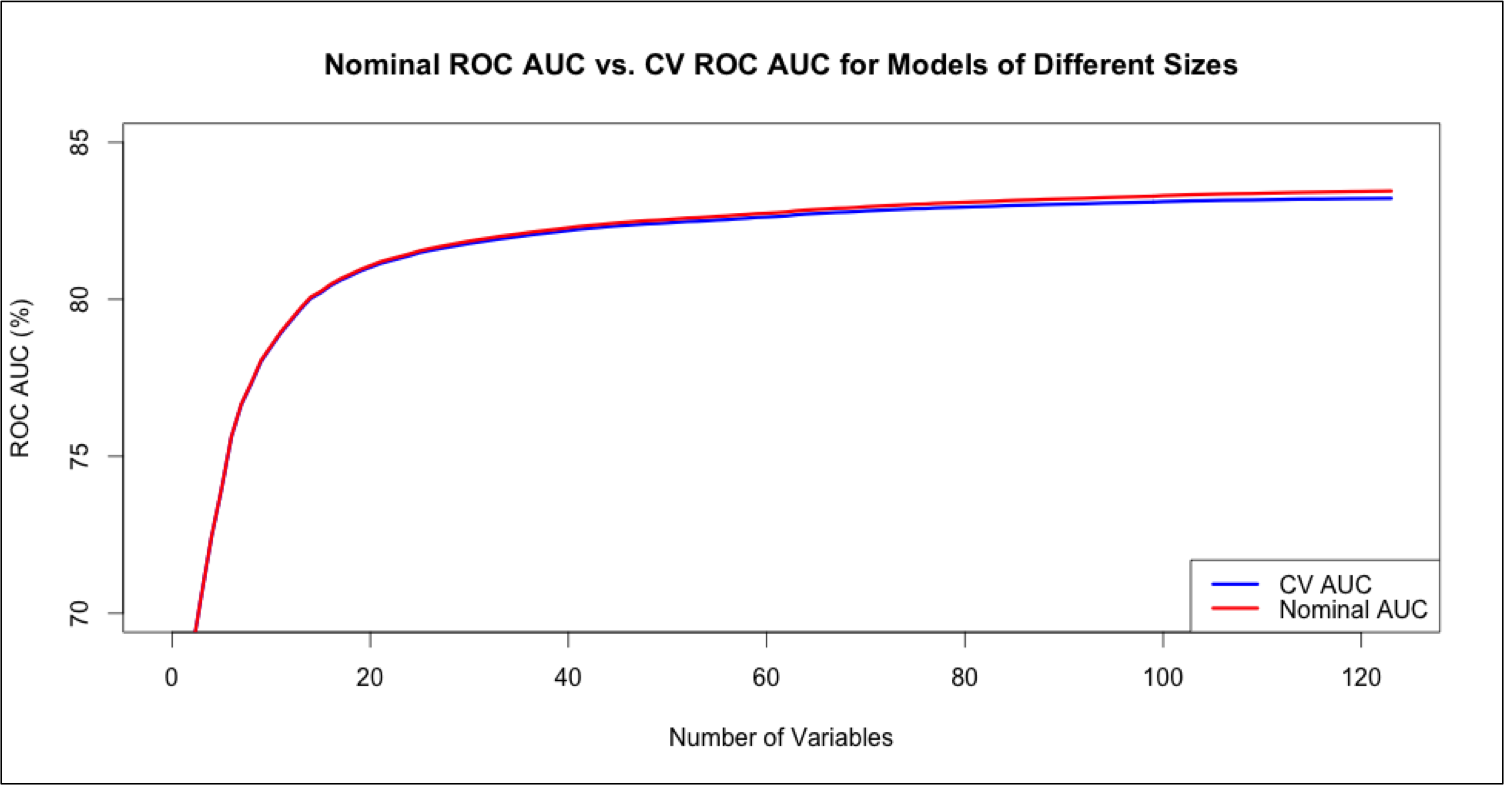}
\end{figure}

We next applied a heuristic threshold, selecting the largest model such that the most recent variable addition increased the cross-validated AUC by at least one basis point (0.01\%). This yields a model containing 53 variables, with an ROC AUC of 82.49\%, or about 1\% lower than a model containing all 684 covariates. 
Our goal is to get an association between $e(\bsx_i)$ and $W_i$.  Additional variables beyond
the 53'rd do not materially improve this association, so we omit them.  

Matching on the propensity score should reduce imbalances on clinically relevant covariates. To evaluate this effect, we use the standardized differences approach of 
\citet[Chapter 9]{rosenbaum2009design}. 
Let $\bar x_{tj}$ and $\bar x_{cj}$ be the treated and control group averages for continuous covariate $j$ in the ODB before matching and let $\hat\sigma^2_{tj}$ and $\hat\sigma^2_{cj}$
be the sample variances within those two groups.
Let $\bar x_{tjk}$ and $\bar x_{cjk}$ be those averages taken over subjects $i\in\odb_k$
and define post-stratification averages as $\tilde x_{tj} = \sum_kn_{ok}\bar x_{tjk}/n_o$
and $\tilde x_{cj} = \sum_kn_{ok}\bar x_{cjk}/n_o$.
These are weighted averages of $x_{ij}$ with greater weight put on observations 
from treatment conditions that are underrepresented in their own strata.
Rosenbaum's standardized differences 
for the original and reweighted data are
$$
\overline\sd_j = \frac{\bar x_{tj}-\bar x_{cj}}{\sqrt{\frac12(\hat\sigma^2_{tj}+\hat\sigma^2_{cj})}}
\quad\text{and}\quad
\wt\sd_j = \frac{\tilde x_{tj}-\tilde x_{cj}}{\sqrt{\frac12(\hat\sigma^2_{tj}+\hat\sigma^2_{cj})}},
$$
respectively. These quantities measure the practical significance of the
imbalance between groups unlike $t$-statistics which have a
standard error in the denominator. Note that Rosenbaum uses the same denominator
in both weighted and unweighted standardized differences.

We first tried ten equal-width propensity score strata to evaluate the standardized differences between treated and control on risk factors listed in \cite{prentice2005combined} before and after adjusting for the propensity score. With the exception of the physical functioning score, all of these covariates were included in the propensity model. Imbalance measures for the continuous covariates can be found in Table \ref{tab:balanceCont}. As we can see, the stratification procedure reduces all standardized differences to less than 0.05 in absolute value, representing very good matches between the populations. As sufficient balance was achieved with ten strata -- and finer stratification would increase variance -- we stuck with ten propensity score strata for this analysis, rather than the twenty used in simulation. 

For categorical variables, the stratification procedure similarly reweights individual women, such that the effective proportion of women in each category changes after stratifying on the propensity score. Standardized differences can also be computed for categorical variables, using the procedure described in \cite{graziano1993research}. We achieve similar balance on two significant categorical variables -- ethnicity and smoking status -- in Tables \ref{tab:balanceRace} and \ref{tab:balanceSmoking}. 


\begin{table}[h]
\caption{\label{tab:balanceCont} Standardized differences (SD) between treated and control populations in the observational dataset, before and after stratification on the propensity score, for clinical risk factors for coronary heart disease.}
\centering
\begin{tabular}{lcccccc}
\toprule
\textbf{}                                                                      & \multicolumn{3}{c}{\textbf{Unweighted}}                                                            & \multicolumn{3}{c}{\textbf{Stratified}}                                                             \\ 
                                                              & \textbf{Test} & \textbf{Ctrl} & \textbf{\begin{tabular}[c]{@{}l@{}}SD\end{tabular}} & \textbf{Test} & \textbf{Ctrl} & \textbf{\begin{tabular}[c]{@{}l@{}}SD\end{tabular}} \\ 
\midrule
\textbf{Age}                                                                   & 60.78         & 64.72            & $-0.56$                                                                       & 63.06         & 63.33            & $-0.04$                                                                       \\ 
\textbf{BMI}                                                                   & 25.55         & 27.11            & $-0.25$                                                                       & 26.71         & 26.62            & $\phm0.00$                                                                        \\ 
\textbf{\begin{tabular}[c]{@{}l@{}}Physical functioning \end{tabular}} & 85.23         & 79.58            & $\phm0.26$                                                                        & 81.15         & 81.23            & $\phm0.03$                                                                       \\ 
\textbf{Age at menopause}                                                      & 50.49         & 50.19            & $\phm0.06$                                                                        & 50.35         & 50.33            & $\phm0.02$                                                                        \\ 
\bottomrule
\end{tabular}
\end{table}

\begin{table}[h]
\caption{\label{tab:balanceRace} Standardized differences (SD) between treated and control populations in the observational database, before and after stratification on the propensity score, for ethnicity category.}
\begin{tabular}{lllllllll}
\toprule
                                                                                           &         & \textbf{White} & \textbf{Black} & \textbf{Latino} & \textbf{AAPI} & \textbf{\begin{tabular}[c]{@{}l@{}}Native \\ American\end{tabular}} & \textbf{\begin{tabular}[c]{@{}l@{}}Missing/\\ Other\end{tabular}} & \textbf{\begin{tabular}[c]{@{}l@{}}SD \end{tabular}} \\ \hline
\multirow{2}{*}{\textbf{\begin{tabular}[c]{@{}l@{}}Before \\ Stratifying \end{tabular}}} & Treated    & 89.0\%         & 2.7\%          & 2.9\%           & 4.0\%         & 0.2\%                                                               & 1.1\%                                                             & \multirow{2}{*}{0.26}                                                \\ 
                                                                                           & Control & 83.1\%         & 8.1\%          & 3.9\%           & 2.8\%         & 0.4\%                                                               & 1.5\%                                                             &                                                                      \\ \midrule
\multirow{2}{*}{\textbf{\begin{tabular}[c]{@{}l@{}}After \\ Stratifying\end{tabular}}}  & Treated    & 83.4\%         & 6.9\%          & 4.3\%           & 3.6\%         & 0.5\%                                                               & 1.4\%                                                             & \multirow{2}{*}{0.05}                                                \\
                                                                                           & Control & 84.8\%         & 6.4\%          & 3.6\%           & 3.4\%         & 0.4\%                                                               & 1.4\%                                                             &                                                                      \\ 
                                                                                           \bottomrule
\end{tabular}
\end{table}

\begin{table}[h]
\caption{\label{tab:balanceSmoking} Standardized differences (SD) between treated and control populations in the observational database, before and after stratification on the propensity score, for smoking category.}
\begin{tabular}{llllll}
\toprule
                                                                                           &         & \textbf{\begin{tabular}[c]{@{}l@{}}Never \\ Smoked\end{tabular}} & \textbf{\begin{tabular}[c]{@{}l@{}}Past \\ Smoker\end{tabular}} & \textbf{\begin{tabular}[c]{@{}l@{}}Current \\ Smoker\end{tabular}} & \textbf{\begin{tabular}[c]{@{}l@{}}SD \end{tabular}} \\ \hline
\multirow{2}{*}{\textbf{\begin{tabular}[c]{@{}l@{}}Before \\ Stratifying\end{tabular}}} & Treated    & 48.7\%                & 46.2\%               & 5.1\%                   & \multirow{2}{*}{0.11}                                                       \\
                                                                                           & Control & 52.3\%                & 41.1\%               & 6.6\%                   &                                                                             \\ \hline
\multirow{2}{*}{\textbf{\begin{tabular}[c]{@{}l@{}}After\\ Stratifying\end{tabular}}}   & Treated    & 50.9\%                & 42.5\%               & 6.6\%                   & \multirow{2}{*}{0.01}                                                       \\ 
                                                                                           & Control & 51.0\%                & 42.7\%               & 6.3\%                   &                                                                             \\ \hline
\end{tabular}
\end{table}

\subsection{Propensity score distribution}
The propensity model can now be applied to the treated and control women in both the observational and RCT populations. Among women in the observational study, the model prediction can be interpreted as the estimated probability of receiving HT given a woman's covariates. Among women in the RCT, the prediction can be interpreted as what that probability \emph{would have been} had the woman not been enrolled in the trial -- or, simply, as a balancing score without intrinsic meaning. 

The distributions of these model predictions can be found in Figure \ref{fig:propHists}. The distribution of propensity scores is indeed quite different between the treated and control populations in the observational data, with women who received HT generally having higher propensity scores than women who did not. No such discrepancy exists in the RCT, where randomization makes it unlikely that the propensity score distribution differs between the treated and control women.

Also notable is the discrepancy between the marginal distributions. Overall, women enrolled in the RCT tended to have a lower probability of receiving HT, with virtually no women in the highest deciles of the propensity score. 

\begin{figure}[h]
\caption{\label{fig:propHists} Propensity score distributions among treated and control women (left panel) and marginal propensity score distributions (right panel) for the observational database (ODB) and randomized controlled trial (RCT) Women's Health Initiative populations.}
\centering
\includegraphics[width = 0.49\textwidth]{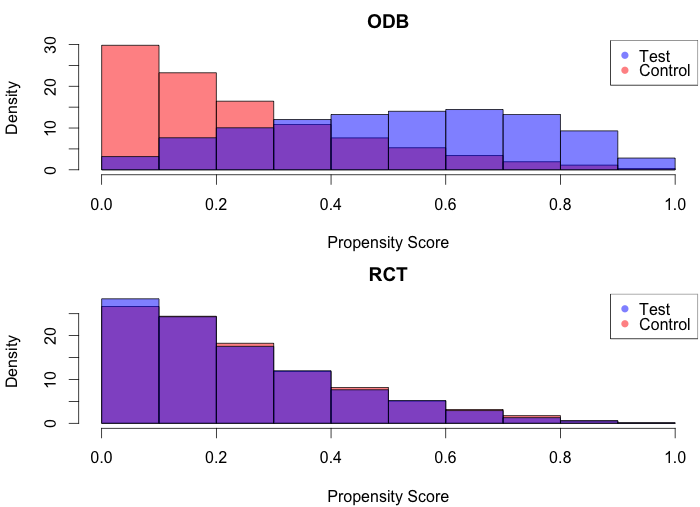}
\includegraphics[width = 0.49\textwidth]{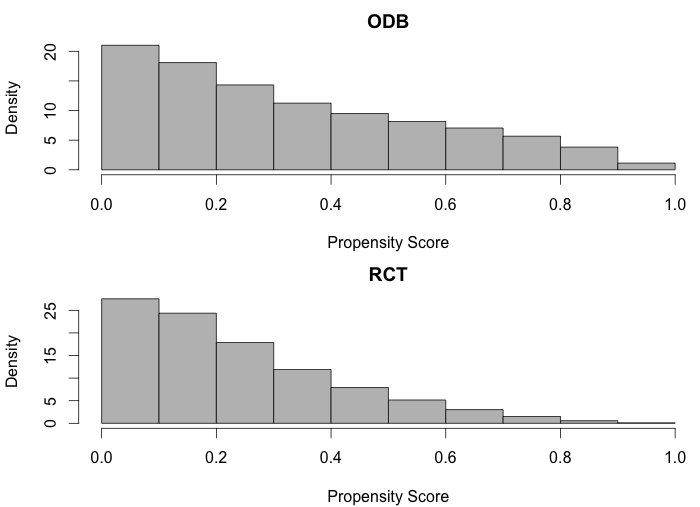}
\end{figure}

\subsection{Gold standard causal effect\label{sec:gsce}}

We now turn our attention to estimation of the ``gold standard" causal effect. We randomly partition the RCT into two subsets of equal size, such that each contains the same number of treated and control women. We select one of these subsets and refer to it as our ``gold" dataset, to be used for estimating the true causal effect. The remaining subset is referred to as the ``silver" dataset, and is used for evaluating our estimators. 

Because of the randomization, we find that treated and control are already well balanced on the coronary heart disease risk factors in the gold dataset, as summarized in Tables \ref{tab:balanceContRCT}, \ref{tab:balanceRaceRCT}, and \ref{tab:balanceSmokingRCT}. 

\begin{table}[h]
\caption{\label{tab:balanceContRCT} Standardized differences (SD) between treated and control populations in RCT gold dataset, for clinical risk factors for coronary heart disease.}
\begin{tabular}{lllr}
\toprule
\textbf{Variable}                                                             & \textbf{Treated} & \textbf{Control} & \textbf{SD} \\ \hline
\textbf{Age}                                                                  & 63.24         & 63.41            & $-$0.02                                                                \\ 
\textbf{BMI}                                                                  & 28.33         & 28.38            & $-$0.01                                                                 \\ 
\textbf{Physical functioning} & 80.97         & 81.11            & $-$0.01                                                                \\ 
\textbf{Age at menopause}                                                     & 44.97         & 46.33            & $-$0.09                                                                \\ \bottomrule
\end{tabular}
\end{table}

\begin{table}[h]
\caption{\label{tab:balanceRaceRCT} Standardized differences (SD) between treated and control populations in RCT gold dataset, for ethnicity category.}
\begin{tabular}{llllllll}
\toprule
                 & \textbf{White} & \textbf{Black} & \textbf{Latino} & \textbf{AAPI} & \textbf{\begin{tabular}[c]{@{}l@{}}Native\\ American\end{tabular}} & \textbf{\begin{tabular}[c]{@{}l@{}}Missing/\\ Other\end{tabular}} & \textbf{\begin{tabular}[c]{@{}l@{}}SD \end{tabular}} \\ \hline
\textbf{Treated}    & 84.1\%         & 6.5\%          & 5.5\%           & 2.1\%         & 0.26\%                                                             & 1.6\%                                                             & \multirow{2}{*}{0.05}                                                \\ 
\textbf{Control} & 84.6\%         & 6.8\%          & 5.1\%           & 1.9\%         & 0.40\%                                                             & 1.2\%                                                             &                                                                      \\ \bottomrule
\end{tabular}
\end{table}

\begin{table}[h]
\caption{\label{tab:balanceSmokingRCT} Standardized differences (SD) between treated and control populations in RCT gold dataset, for smoking category.}
\begin{tabular}{lllll}
\toprule
& \textbf{\begin{tabular}[c]{@{}l@{}}Never  Smoked\end{tabular}} & \textbf{\begin{tabular}[c]{@{}l@{}}Past Smoker\end{tabular}} & \textbf{\begin{tabular}[c]{@{}l@{}}Current Smoker\end{tabular}} & \textbf{\begin{tabular}[c]{@{}l@{}} SD \end{tabular}}\\ \hline
\textbf{Treated}    & 50.1\%         & 38.7\%          & 11.2\%   & \multirow{2}{*}{0.03}                                                \\ 
\textbf{Control} & 50.6\%         & 39.1\%          & 10.2\%                                                      &                                                                      \\ \bottomrule
\end{tabular}
\end{table}

We face two key question in estimating the gold standard causal effect: against whom do we compare each of the treated units, and how do we reweight the causal effects? We want to impose minimal assumptions on the data and also make use of the wealth of prior research from WHI. Hence, for comparisons, we take the approach of \cite{prentice2005combined} and make treated-vs-control comparisons among women in the same WHI component (clinical trial or observational study) and same five-year age category. Since we are not using time-to-event methods, we do not stratify on length of enrollment but rather compare outcomes across the entire follow-up period, which averages 8.19 years for treated women and 8.08 years for control women. Women in the control group were observed for about $1.3$\% less time, a difference that is too small to matter here.

For reweighting, we opt to use the maximum entropy reweighting approach described by \cite{Hartman_fromsate} -- again, trying to reduce the assumptions we place on the data. We select the same set of confounding factors used in \cite{prentice2005combined} as our variables on which to conduct maximum entropy reweighting. The authors computed separate hazard ratios for different levels of these variables, exemplifying the belief that the causal effect may vary based on them. 

A naive difference between the frequency of CHD events in the treated versus control populations of the RCT gold dataset yields a causal effect of 0.67\%. Age stratification and maximum entropy reweighting, as implemented using the \texttt{ebal} package \citep{ebal}, yield substantially similar quantities. This indicates that combined hormone therapy increases the frequency of coronary heart disease. 
Inspection reveals that the women in the RCT have higher BMIs, are more likely to smoke, and went through menopause younger, all of which have a negative interaction effect with HT on coronary heart disease. 




\subsection{Prognostic modeling}\label{subsec:outcomeModeling}

At the end of Section~\ref{sec:stratification} we mentioned a strategy
of subdividing propensity bins on a prognostic measure that predicts
potential outcomes, in order to make Assumption~\ref{assumption:strongtreatmenteffect} more reasonable.
The spiked-in estimator works best for small values of $\Delta_k$, that is
when the mean potential outcome vector $(Y_{it},Y_{ic})$ is
nearly the same for both ODB and RCT subjects within each stratum.
In this section we develop such a stratification for the WHI data using a prediction of CHD.
We subdivide the propensity based strata by a predicted value for $\e(Y_{ic})$ derived
from the ODB control population. We refer to this as the ``prognostic score."



Since the CHD outcome is binary, we use logistic regression to predict $Y_{ic}$.
We use the same 684 covariates obtained earlier, and fit the model via forward stepwise selection, allowing a maximum of 25 variables into the model (this cap is imposed purely to speed up computation). The model coefficients are derived using only the data from the observational control population. Selecting only control women for our training set is sensible, since our goal is to estimate $Y_{ic}$; we use the ODB only because this dataset is sufficiently large for stable coefficient estimation. 
The model achieves a 77.8\% (cross-validated) AUC in the ODB control population.
It also has a 73.7\% AUC in the RCT gold control population, though we would not see that in an application.
See Figure \ref{fig:outcomeROCAUCs}. 


To evaluate comparability between the RCT and ODB data before and after stratification on the prognostic score, we again make use of standardized differences. We separately evaluate the treated and control groups, looking at standardized differences between the RCT and ODB populations with CHD as the outcome. Strata are weighted according to the proportion of the ODB population within them. 

Because the CHD outcome is rare, the distribution of prognostic scores is highly right-skewed. We thus use equal-depth stratification, rather than the equal-width stratification used for the propensity score. Since we are already using ten propensity score strata, we seek to sub-stratify on relatively few additional prognostic score bins so as not to reduce sample sizes too dramatically. 
 
We find that prognostic score stratification does not yield an improvement in comparability. But if we use three prognostic score strata, the standardized differences are nonetheless within acceptable ranges with and without outcome score stratification. Among control women, we obtain a standardized difference of $-0.02$ when stratifying on the propensity score only. The difference grows slightly when we additionally stratify on the outcome score to $-0.03$. The analogous numbers for the treated women are $-0.09$ when stratifying only on the propensity score, and $-0.10$ when stratifying on both. 

\begin{figure}[h]
\caption{\label{fig:outcomeROCAUCs} ROC AUC scores for logistic regression outcome model in the control populations of the ODB and RCT silver datasets. }
\centering
\includegraphics[width = 0.95\textwidth]{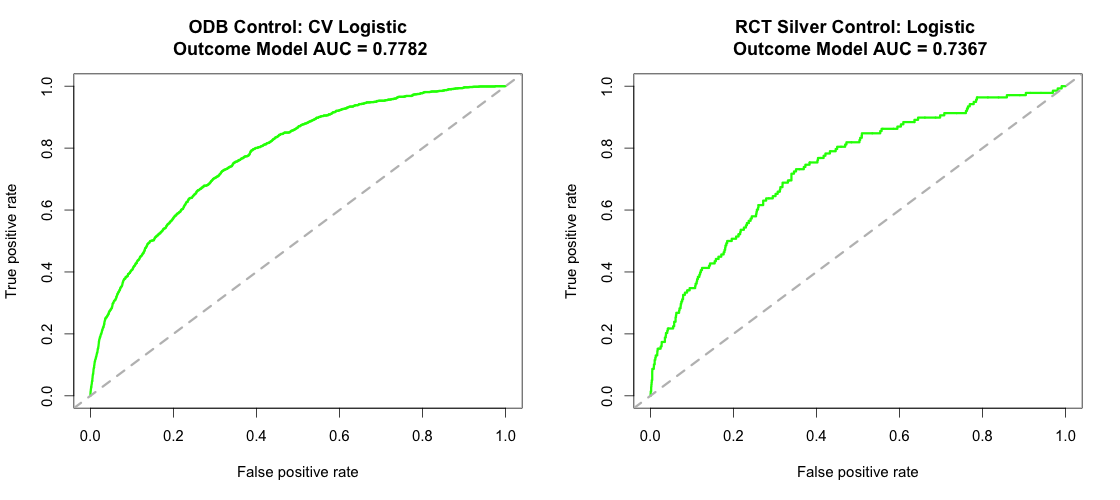}
\end{figure}

\subsection{Results}\label{subsec:results}
We compared these methods under two scenarios: a small RCT of just $1,000$ units; and a full-sized RCT the size of our silver dataset ($8,269$ units). To induce variation, we draw $100$ bootstrap replicates, resampling the entire ODB and resampling RCTs of each size. In each replicate we computed all of our estimators of the causal effect of HT on CHD.

Boxplots of those estimates appear in Figure~\ref{fig:bootstrapComparisons}, with the ``gold standard" age-stratified, reweighted effect of 0.26\% from Section \ref{sec:gsce} given by a horizontal gold line. 
There is a bias-variance tradeoff among these methods.
Figure~\ref{fig:MSEs} shows their root mean square errors with respect to the gold standard estimate.
At both sample sizes, the spiked and dual-spiked estimates come out best. 

\begin{figure}[h]
\caption{\label{fig:bootstrapComparisons} 
Causal estimators computed over 100 bootstrap replicates for small and larger RCT sizes. The three estimators on the left are all computed from a single dataset, while the right three estimators make use of both the RCT and ODB data. The dashed gold line is our gold standard estimate.  }
\centering
\includegraphics[width = 0.9\textwidth]{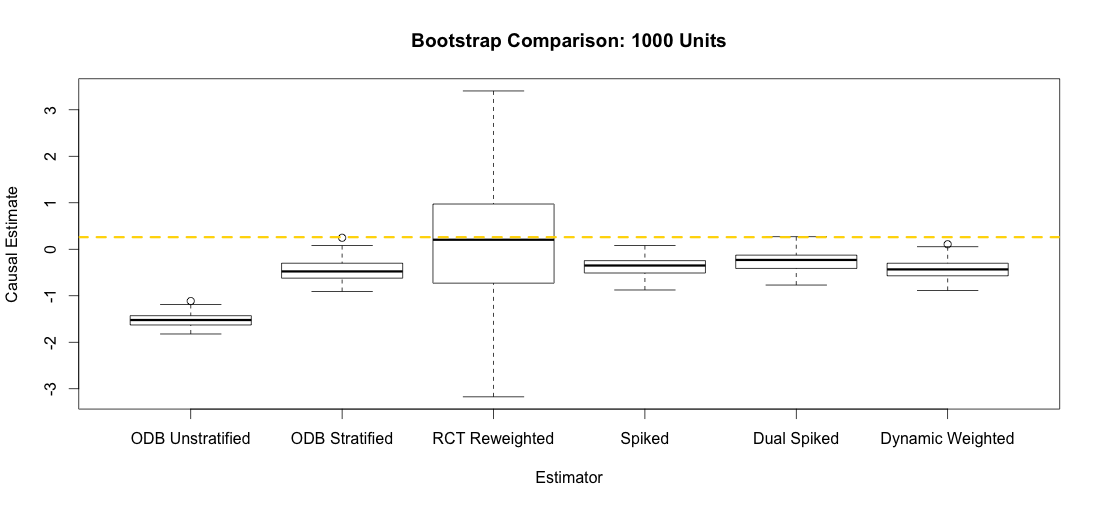}\\
\includegraphics[width = 0.9\textwidth]{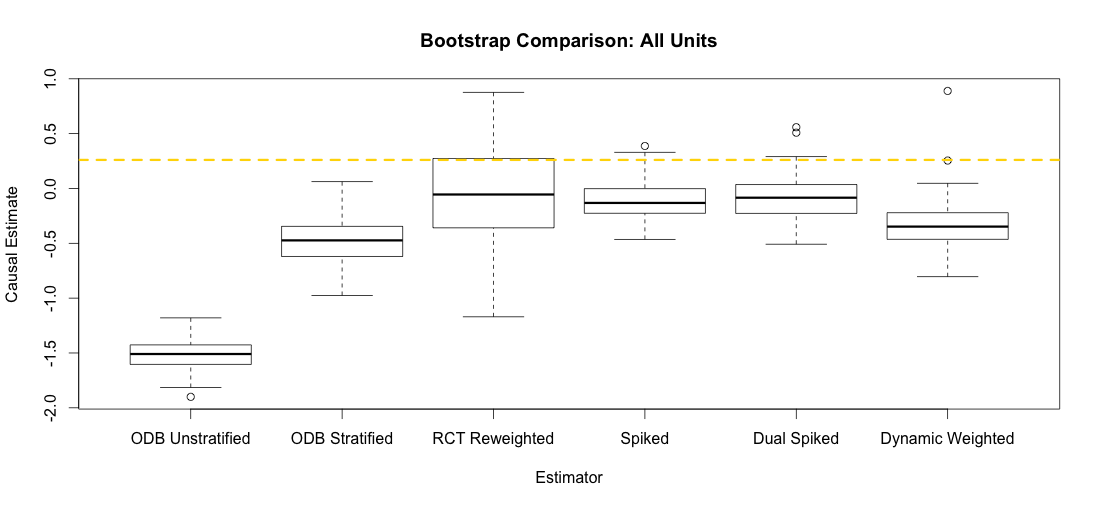}
\end{figure}

The RCT-only estimator has small bias but very large variance.
This stems from the small size of the RCT silver dataset: just over 8,000 women vs.\ over 50,000 in the ODB dataset.
The naive ODB estimator (unstratified) has the smallest variance but very large bias giving it the largest RMSE of all the methods.
This stems from the selection bias we have discussed.
Stratifying the ODB data on the propensity score (using 10 equal-width bins), corrects for much of the
bias but leaves it with greater bias than the hybrid methods.
The two spiked  estimators have relatively low bias and also small variance across the bootstrap replicates. 
At the smaller sample size, the dual spiked estimator does best, owing to lower bias than the spiked estimator. At the larger sample size, the dual-spiked estimator still has smaller bias, but its variance is larger, such that the two estimators perform about the same. 

The dynamic weighted estimator is not competitive with the two variants of the spiked estimator, owing to higher variance and much higher bias.
This estimator is more robust to distributional dissimilarity between the RCT and ODB populations, which, as discussed in the prior section, does not appear to be an issue in this dataset. 


\begin{figure}[h]
\caption{\label{fig:MSEs} 
Root mean square error when estimating the causal effect of HT on CHD, across 100 bootstrap replicates for small and large RCT sizes. The gold standard causal effect is taken to be the age-stratified reweighted estimator, the magnitude of which is shown via the dashed gold line. }
\centering
\includegraphics[width = 0.9\textwidth]{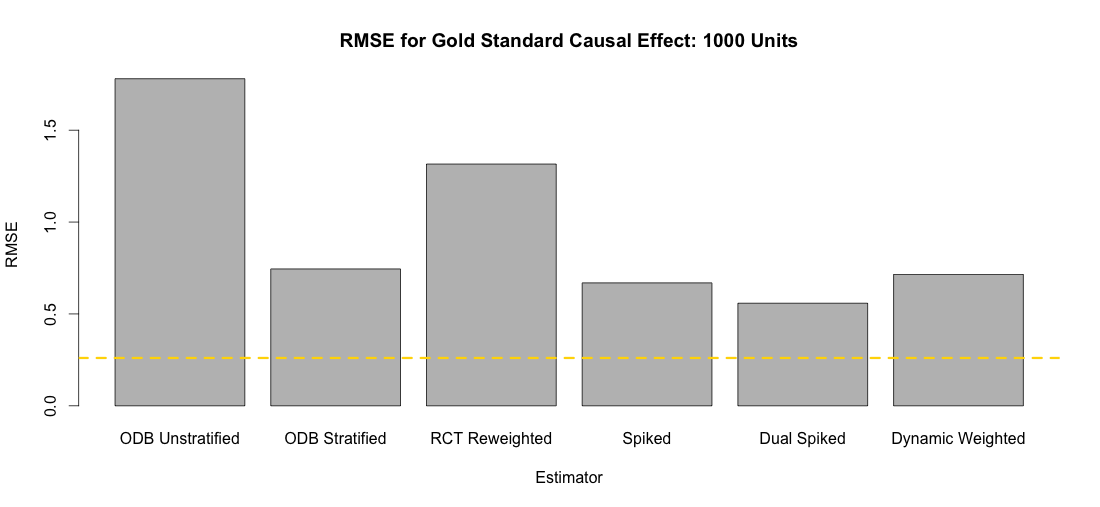}\\
\includegraphics[width = 0.9\textwidth]{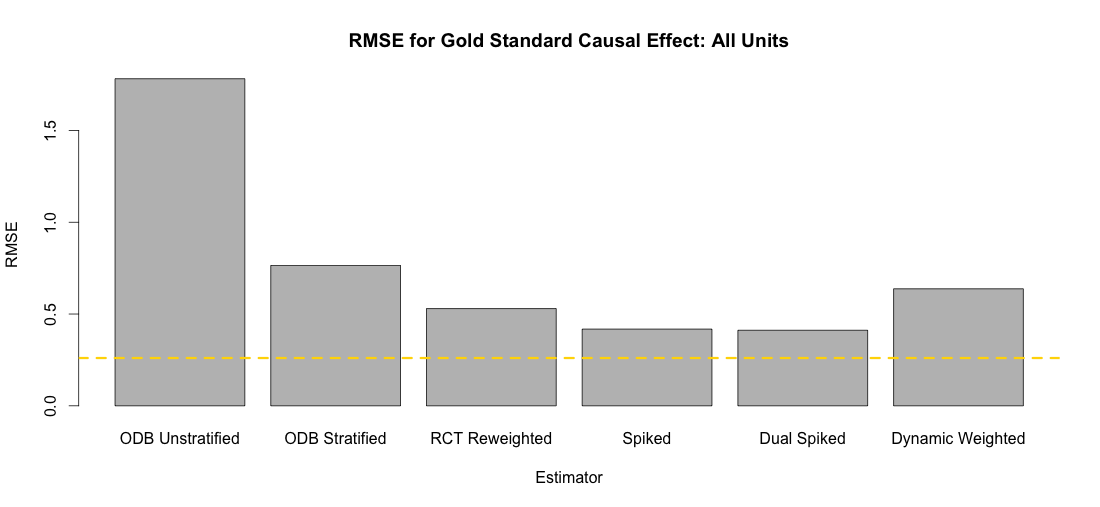}
\end{figure}

\section{Conclusions}\label{sec:conclusions}

We have developed propensity based methods to merge data from a randomized
controlled trial with data on the same phenomenon from an observational data base.
Our goal is to reduce the root mean squared error in an estimate of the overall population treatment effect.

The strategies we use are based on the propensity that the RCT data would have
had, had they been in the ODB.  The simplest strategy is to spike the RCT data into
the corresponding propensity strata.  It works well in theory and in experiments
when the covariate distribution in the RCT matches that of the ODB.
If however those distributions differ sharply, as they could for an RCT with
restrictive enrollments, then the spiked-in estimator can perform very badly
and even be worse than using the ODB alone without the RCT.
We developed an alternative estimator based on taking a weighted
average of the ODB and RCT data within every stratum.   
An oracle knowing the biases and variances of the ODB and RCT within
each stratum could make a principled choice of weight vector.
We developed an estimator that uses the plug-in principle to
estimate that weight vector.  On biased examples, it greatly
outperformed both the spike-in estimator and the ODB itself.

We lastly evaluated our estimators on the Women's Health Initiative dataset, with the goal of obtaining the causal effect of hormone therapy on coronary heart disease. We had access to a large observational study, as well as a smaller randomized controlled trial. Half the RCT data was used to estimate a gold standard causal effect, and a propensity model was fit to the ODB data to correct for biases due to confounding. We also generated a variant of the spiked-in estimator, termed the ``dual spiked" estimator, which makes use of stratification on both the propensity score and an ``outcome score" predicting the likelihood of CHD in an untreated woman. Estimators were evaluated based on their RMSE for the gold standard causal effect across 100 bootstrap replicates. The dual spiked estimator performed best, followed closely by the spiked-in estimator. These estimators introduced a small amount of bias, but greatly reduced the variance across bootstrap replicates. 

Our conclusion is that the spiked and dual-spiked estimator are the best choices in many practical examples, as long as the RCT and ODB outcome distributions are similar conditional on the propensity score. Otherwise, we prefer the dynamic weighted estimator. 

\section*{Acknowledgments}
Evan Rosenman was supported by the Department of Defense (DoD) through the National Defense Science \& Engineering Graduate Fellowship (NDSEG) Program. This work was also supported by the NSF under grants DMS-1521145 and DMS-1407397.

\bibliographystyle{apalike} 
\bibliography{rct+odb}      


\section*{Appendix}

This appendix contains two of the lengthier proofs. 

\subsection*{Proof of Theorem \ref{theorem:general}}\label{sec:proofthmgeneral}

First let $x_t = \sum_{i\in\sets}\wit$ and $y_t = \sum_{i\in\sets}\wit\yit$.
From Proposition~\ref{prop:deltaratio},
\begin{align*}
\edelt\Bigl(\frac{y_t}{x_t}\Bigr)
& = \frac{y_{t,0}}{x_{t,0}} -\frac{\cov( y_t-\rho_t x_t,x_t)}{x_{t,0}^2}
\end{align*}
with $x_{t,0}=\e(x_t)$, $y_{t,0}=\e(y_t)$ and $\rho_t = y_{t,0}/x_{t,0}$.
Next $x_{t,0}=np_t$ and $y_{t,0} = n(s_t+\mu_tp_t)$.
Therefore
$\rho_t = 
\mu_t+{s_t}/{p_t}.$
By independence of the $\wit$,
\begin{align*}
\cov(y-\rho x,x) &=\sum_{i\in\sets} \cov( \wit(\yit-\rho),\wit)
=\sum_{i\in\sets} (\yit-\rho)p_i(1-p_i) \\
&= - \sum_{i\in\sets} (\yit-\rho)p_i^2 =O(n),
\end{align*}
because $|\yit|\le B$. Furthermore, $x_{t,0}^2\ge\epsilon^2n^2$ and so
$$
\edelt\Bigl(\frac{y_t}{x_t}\Bigr)=
\mu_t + \frac{s_t}{p_t}+O\Bigl(\frac1n\Bigr).
$$
Applying the same argument to the second term in $\hat\tau$
establishes equation~\eqref{eq:genedelt} for $\edelt(\hat\tau)$.

Now we turn to the delta method variance, using Proposition~\ref {prop:deltaratiodiff}.
By independence of $W_i$,
\begin{align*}
\var(y_t-\rho_tx_t) &= 
\sum_{i\in\sets}\var( W_i(\yit-\rho_t))
=\sum_{i\in\sets} p_i(1-p_i)(\yit-\rho_t)^2=nS_{tt}(\sets)
\end{align*}
and by the same argument,
$\var(y_c-\rho_cx_c) = nS_{cc}(\sets)$.
Next
\begin{align*}
\cov(y_t-\rho_tx_t,y_c-\rho_cx_c) &= 
\sum_{i\in\sets}\cov( W_i(\yit-\rho_t),(1-W_i)(\yic-\rho_c))=-nS_{tc}(\sets)
\end{align*}
because $\cov(W_i,1-W_i)=-p_i(1-p_i)$.
The denominators in  Proposition~\ref {prop:deltaratiodiff}
simplify to $n^2p_t^2$, $n^2p_c^2$, and $n^2p_tp_c$.
Then
$$
\vdelt\Bigl(
\frac{y_t}{x_t}-\frac{y_c}{x_c}
\Bigr)
=\frac1n\Bigl(
\frac{S_{tt}}{p_t^2}+\frac{S_{cc}}{p_c^2}
-2\frac{-S_{tc}}{p_tp_c}
\Bigr)
$$
completing
the proof of~\eqref{eq:genvdelt}.

\subsection*{Proof of Corollary \ref{cor:rctMoments}}\label{sec:proofcorrctMoments}
The RCT sampling probabilities are all $p_r$. 
Theorem~\ref{theorem:general} applies with $\epsilon=\min(p_r,1-p_r)$. 
Because the sampling probability is the same for all subjects $i$,
the covariances $s_t(\rct_k)$ and $s_c(\rct)$
from equation~\eqref{eq:prcovfors} vanish,  making $\edelt(\hat\tau_{rk})=O(1/n_{rk})$. 
Next from Theorem~\ref{theorem:general},
\begin{align*}
\vdelt(\hat \tau) 
&= \frac{1}{n_{rk}} \left( \frac{S_{tt}}{p_t^2} + \frac{S_{cc}}{p_c^2} + 2 \frac{S_{tc}}{p_tp_c} \right) 
\end{align*}
with parts defined using $\sets=\rct_k$. Then $p_t=p_r$, $p_c=1-p_r$,
$\rho_t=\mu_t$, $\rho_c=\mu_c$,
$S_{tt} = [p_r(1-p_r)/n_{rk}]\sum_{i\in\rct_k}(\yit-\mu_t)^2$,
$S_{cc} = [p_r(1-p_r)/n_{rk}]\sum_{i\in\rct_k}(\yic-\mu_c)^2$, and 
$S_{tc} = [p_r(1-p_r)/n_{rk}]\sum_{i\in\rct_k}(\yit-\mu_t) (\yic-\mu_c)$. 
Making these substitutions,
\begin{align*}
\vdelt(\hat \tau) 
&= \frac{1}{n_{rk}} \left( \frac{S_{tt}}{p_t^2} + \frac{S_{cc}}{p_c^2} + 2 \frac{S_{tc}}{p_tp_c} \right) \\
&=\frac{p_r(1-p_r)}{n_{rk}^2}\left(\sum_{i\in\rct_k} \frac{(\yit-\mu_t)^2}{p_r^2}  + \frac{(\yic-\mu_c)^2}{(1-p_r)^2} + 2\frac{(\yit-\mu_t)(\yic-\mu_c)}{p_r (1-p_r)} \right) \\
&=\frac{p_r(1-p_r)}{n_{rk}^2}\sum_{i\in\rct_k} \left( \frac{(\yit-\mu_t)(1-p_r) + (Y_{ic} - \mu_c)p_r}{p_r(1-p_r)}  \right)^2 \\
&=\frac{\bar\sigma^2_{rk} }{p_r (1-p_r) n_{rk}}
\end{align*}
where
$\bar\sigma^2_{rk}=(1/n_{rk})\sum_{i\in\rct_k} [ (\yit-\mu_t)(1-p_r)+(\yic-\mu_c)p_r]^2$.
Under Assumption~\ref{assumption:strongtreatmenteffect},
$\yit-\mu_t=\yic-\mu_c$ and $\bar\sigma^2_{rk}$ simplifies as given. 
Similarly, substituting $p_r=1/2$ yields the other given simplification.	








\end{document}